\documentclass[letterpaper, 10 pt, conferenceb]{ieeeconf}
\usepackage{mathrsfs}
\usepackage{amssymb}
\usepackage{mathtools}
\usepackage{amsmath}
\usepackage{nccmath}
\usepackage{graphicx}
\usepackage{caption}
\usepackage{pifont}

\usepackage{comment}
\usepackage{cite}

\usepackage{algorithm}
\usepackage{algpseudocode}

\usepackage{caption}
\usepackage{subcaption}

\newtheorem{lemma}{Lemma}
\newtheorem{proposition}{Proposition}
\newtheorem{definition}{Definition}
\newtheorem{remark}{Remark}

\newtheorem{conjuncture}{Conjuncture}

\ifodd 0
\newcommand{\rev}[1]{{\color{blue}{#1}}}
\else
\newcommand{\rev}[1]{#1}
\fi

\begin{document}

\title{The Impact of Network Design Interventions on the Security of Interdependent Systems}
\author{Pradeep Sharma Oruganti, Parinaz Naghizadeh, and Qadeer Ahmed
}

\maketitle

\begin{abstract}
We study the problem of defending a Cyber-Physical System (CPS) consisting of interdependent components with heterogeneous sensitivity to investments. In addition to the optimal allocation of limited security resources, we analyze the impact of an orthogonal set of defense strategies in the form of network design interventions in the CPS to protect it against the attacker. We first propose an algorithm to simplify the CPS attack graph to an equivalent form which reduces the computational requirements for characterizing the defender's optimal security investments. We then evaluate four types of design interventions in the network in the form of adding nodes in the attack graph, interpreted as introducing additional safeguards, introducing structural redundancies, introducing functional redundancies, and introducing new functionalities. 
We identify scenarios in which interventions that strengthen internal components of the CPS may be more beneficial than traditional approaches such as perimeter defense. We showcase our proposed approach in two practical use cases: a remote attack on an industrial CPS and a remote attack on an automotive system. We highlight how our results closely match recommendations made by security organizations and discuss the implications of our findings for CPS design. 
\end{abstract}



\section{Introduction}
\label{sec:introduction}

Vulnerabilities in modern Cyber-Physical Systems (CPS) are increasingly exploited by attackers to launch sophisticated attacks on their safety-critical components. Automation, interdependence between assets in a network, and connectivity between different networks, all complicate the task of protecting the many assets within a CPS. Further, modern attacks are initiated and choreographed over multiple assets in the network, with the attackers remaining undetected for long stretches of time as they work their way to the most critical targets \cite{falliere2011w32,greenberg2015hackers,slay2007lessons}. In response, CPS operators need to decide on an optimal allocation of their often limited security resources throughout a network by taking into account the functionality and security attributes of different components. 

Given the conflicting goals of the attacker and the CPS operator, game-theoretic modeling and analysis can be used to 
provide insights and recommendations for the operators' optimal security decisions. In particular, there has been significant work on security games on networks for attack detection and improving network resilience \cite{amin2013security,  nguyen2017multi, liu2010security, hota2018game, smith2020survey, zeng2019survey, milovsevic2019network, pirani2021game, pirani2021strategic, abdallah2020behavioral,  milani2020harnessing, oruganti2021impact}. Several of these works have used ``attack graph'' models to study attacks on interconnected CPS. The motivation for these models is that, to successfully compromise targets internal to the network, attackers generally initiate stepping-stone attacks from external nodes, and gradually work their way to the critical assets. As such, the nodes in the attack graph are used to represent the CPS assets, while the connectivity between them shows all the components that an attacker needs to (sequentially) compromise in order to reach the CPS's most critical assets. In this paper, we similarly use an attack graph model to analyze how a CPS defender can optimally deploy its security resources to best protect the CPS against an attacker. 

\subsection{Contributions and paper overview}\label{sec:contributions}
We present two main extensions over existing works that have used an attack graph formalization to study CPS security: analyzing optimal security investments when assets have heterogeneous return-on-investment, and assessing the impacts of network design interventions. We detail each of these extensions, along with the main analytical and practical implications of considering them. 

\subsubsection{Modeling assets' return-on-investment} First, we extend the attack graph models studied in prior works (e.g.~\cite{hota2018game, abdallah2020behavioral, oruganti2021impact}) by introducing a \emph{return-on-investment} feature, $\kappa_i$, for each asset $v_i$. This term, which is heterogeneous across assets, can be thought of as the rate of decrease in that asset's security risk per unit of investment. This captures realistic scenarios in which investing in some assets can provide better ``bang for the buck''. To the best of our knowledge, an attack graph with non-uniform node sensitivities has only been considered in \cite{abdallah2021morshed} but with a primary focus on numerical experiments. Our work therefore extends this literature by introducing nodes' sensitivities to investments and providing an analytical study of the resulting games.

In particular, in Section~\ref{sec: reductions}, we present an algorithm for transforming the attack graph of the resulting security game into an ``equivalent'' reduced form graph which considerably simplifies the computation load of identifying the optimal security investments and assessing the expected loss of the network (this is achieved by reducing the number of decision variables and constraints in the underlying minmax optimization problem). We further show that the resulting equilibrium investment strategies may recommend spreading investments on assets internal to the network; this is contrary to previous results in the homogeneous return-on-investment model which could only identify perimeter defense (as opposed to strengthening internal assets) and ``min-cut'' strategies (as opposed to spreading investments) as optimal for the defender \cite{oruganti2021impact, abdallah2020behavioral}. 

\subsubsection{Assessing the impacts of design interventions} Our second contribution is to analyze an orthogonal set of defender actions in the form of \emph{network design interventions}. In particular, in addition to optimally allocating her security budget, the defender can choose to modify the CPS by adding new nodes in the attack graph (as detailed shortly). To the best of our knowledge, the only other work considering network design interventions is \cite{milani2020harnessing} which looks at hiding or revealing edges of an attack graph to change an attacker's perception, while the original network is not modified. 

Specifically, we focus on four possible re-design actions that can be taken by a CPS operator, which result in the introduction of additional nodes in the attack graph: 
\begin{itemize}
    \item[(a)] Adding a node in series with existing nodes in the graph. Examples include adding an encryption device, or requiring stronger passwords. 
    \item[(b)] Adding a node in parallel with an existing node. Examples include adding an additional user to the CPS. 
    \item[(c)] A hybrid case of simultaneously adding a series and a parallel node to an existing node. Examples include adding an additional sensor to provide redundant information for anomaly detection. 
    \item[(d)] Adding additional input nodes. Examples include introducing an additional functionality in the system, such as adding Bluetooth connectivity to a device. 
\end{itemize} 
In Section~\ref{sec: formations}, we consider each of these interventions when applied to a base network. We find the equilibrium outcomes of the security game on the modified attack graph, and compare the resulting expected network losses against that of the base network to elaborate on the security implications of each design intervention. 

\subsubsection{Numerical experiments and practical implications} In Section~\ref{sec: application}, we illustrate both our attack graph reduction algorithm and our proposed design interventions in two (numerical) use cases: a remote attack on an industrial SCADA system  and a remote attack on an automotive system. We also discuss the practical implications of our findings. 
In particular, for the SCADA system, we compare our recommended investment strategy against a perimeter defense strategy. The strategy recommended by our approach outperforms the perimeter defense strategy which matches current trends in industry practice \cite{bissell2019cost}. Further, in our analysis of the automotive system which follows the penetration testing report \cite{cai20190}, we find that our findings closely match the countermeasures recommended by security agencies. These observations indicate the potential value of our proposed framework as an analytical tool to help in strategic decision-making. 

\subsection{Related work}

Our work is within the literature on using an attack graph formalization in the study of CPS security \cite{nguyen2017multi, liu2010security, hota2018game, abdallah2020behavioral,  milani2020harnessing, oruganti2021impact, abdallah2021morshed}, as detailed in Section~\ref{sec:contributions}. Complimentary to these models, there exists a rich literature on \emph{network interdiction games} as an alternative approach to the study of optimal resource allocation in networks (see \cite{smith2020survey} for a survey). In general, the attacker in a network interdiction game aims to identify the shortest path from the source nodes to the target assets, and the defender's security investments are aimed at ``lengthening of the arcs'' in the attack graph to thwart the attacker. One difference between network interdiction and attack graph formulations is that the former models do not tend to capture intermediate losses from the traversed assets in the attack path, i.e., a loss is incurred only when the attacker reaches the target asset through its selected (shortest) path. In contrast, an attack graph formulation allows us to model the loss from the intermediate assets, with a loss incurred even if the attacker only manages to partially progress through an attack path. We further compare our attack graph model parameters with those in network interdiction games in Section~\ref{subsec: secgame}.  

Optimal cyber-risk management and security resource allocation has also been studied using concepts from Probabilistic Risk Analysis (PRA) in \cite{pate2021probabilistic,smith2018cyber, pate2018cyber}. The networks considered in these works are different from stepping-stone attack graphs, in that attacks may be targeted at any individual node directly. Attack graphs models based on Bayesian graphs and Markov chains have also been used to study the overall vulnerability of IT systems in \cite{xie2010using, liu2010security, abraham2015exploitability}; however, these works do not consider design interventions or optimal investment decisions. 

An earlier version of our work appeared in \cite{oruganti2021impact} for the homogeneous return-on-investment model. We extend \cite{oruganti2021impact} by generalizing our reduction algorithm and design interventions to account for heterogeneous returns-on-investment. We show that the recommendations from our new model outperform the perimeter defense strategies recommended by \cite{oruganti2021impact} when asset sensitivities are taken into account, and use two new case studies to show that our findings are close to realistic manufacturer decisions and security agencies' recommendations. 
\section{The Security Game Framework}
\label{sec: framework}

\subsection{The attack graph}
\label{subsec: atkgrf}
We consider a cyber-physical system (CPS) modeled as an \rev{acyclic directed attack graph} $\mathcal{G} = \{\mathcal{V}, \mathcal{E}\}$, where $\mathcal{V}$ represents the set of nodes and $\mathcal{E}$ represents the set of edges of the graph. A directed edge $(v_i, v_j) \in \mathcal{E}$ connecting node $v_i$ to node $v_j$ indicates that an attack on $v_j$ can be launched once $v_i$ is compromised. The attacks can be initiated from any of the outermost \emph{entry} or \emph{source} nodes of the graph, and are aiming to reach the \emph{target} or \emph{goal} asset. The set of entry nodes is represented as $\mathcal{V}_s \subseteq \mathcal{V}$ and the unique target node is $v_g \in \mathcal{V}$. 

A path $P_{ij}$ between nodes $v_i$ and $v_j$ is a sequence of connected nodes $\{v_i, v_{i+1}, \ldots, v_j\}$, i.e. $P_{ij}$ = $\{v_i \rightarrow v_{i+1} \rightarrow \cdots \rightarrow v_j\}$; let $\mathcal{P}_{ij}$ denote the set of all such paths. All the nodes that can be reached from a node $v \in \mathcal{V}$ (through one or more steps \rev{and including $v$ itself}) are denoted as $Post(v)$, and all the nodes from which $v$ \rev{(including itself)} can be reached are denoted $Pre(v)$. Each node $v_i$ is endowed with a stand-alone loss (financial or functional) $L_i \geq 0$, incurred if the node is successfully compromised. We assume $L_g > 0$, where $L_g$ is the loss associated with the target asset $v_g$. Table~\ref{t:notation} in Appendix~\ref{app:notation} 
summarizes our notation.

\subsection{The security game}
\label{subsec: secgame}
We study a Stackelberg game between an attacker and a defender. The defender acts first by 
deploying defense resources over the nodes of $\mathcal{G}$. Let $x_i\in \mathbb{R}_{\geq 0}$ denote the security investment on node $v_i$. We assume that given an investment $x_i$, the probability of successful attack on node $v_i$ is given by:
\begin{equation}
	p_i(x_i) = p_i^0 e^{-\kappa_ix_i}
\end{equation}
Here, $p_i^0 \in (0, 1]$ denotes the default probability of compromise under no investment, and $\kappa_i \geq 1$ is a node's sensitivity to investments (with higher $\kappa$ indicating higher marginal benefit-on-investment). {Similar models have been considered in prior works \cite{hota2018game, abdallah2020behavioral, oruganti2021impact} when $\kappa_i=1, \forall ~i$. Our work extends these works by introducing node sensitivities and providing an analytical study of the resulting games.}\footnote{{Similar elements appear in shortest path network interdiction game formulations~\cite{smith2020survey}. Specifically, an arc $(v_i, v_j)$ in those models has length $c_{ij}+d_{ij}x_{ij}$ given the (typically binary) interdiction decision $x_{ij}$. The parameters $p_i^0$ and $\kappa_i$ in our model are similar to $c_{ij}$ and $d_{ij}$ in such models.}} 

\rev{We consider a game of full information, i.e., the attacker and defender both have knowledge of the network topology, all node attributes, and each other's utility functions and action sets.} The attacker's action consists of selecting one path $P_{sg} \in \mathcal{P}_{sg}$ to initiate a sequence of attacks starting from some $v_s \in \mathcal{V}_s$ with the objective to reach and compromise the target node $v_g$. Assuming a worst-case attacker, its goal is to identify the path to perform stepping-stone attacks that would lead to the maximum expected loss on the CPS. In response, the defender chooses an investment profile $\mathbf{x}=[x_1, x_2, \ldots, x_{|\mathcal{V}|}]$ to minimize the loss in face of such attacker. 

Formally, the defender solves the following problem: 
\begin{equation}
\begin{gathered}
    \min_{\mathbf{x}} \max_{P_{sg}\in\mathcal{P}_{sg}} ~ \sum_{v_i \in P_{sg}}L_i  \prod_{v_j\in \rev{P_{sg}}\cap Pre(v_j)} p_j(x_j)~ \\
    \text{s.t.}\quad \sum_{i = 1}^{{|\mathcal{V}|}} x_i \leq B~, \text{ and }~ x_i \geq 0,~ \forall i\in\mathcal{V}~.
\end{gathered}
\label{eq: atkOptim}
\end{equation} 
Here, $B$ is the security budget available to the defender. We use $\mathbf{x}^*$ and $\mathbf{L}^*$ to denote the optimal solution of \eqref{eq: atkOptim} and the expected loss under this investment profile, respectively.  
The solution to \eqref{eq: atkOptim} determines the Stackelberg equilibrium strategies for the defender. \rev{We note that the objective function is strictly convex, and the feasible region is non-empty and compact; therefore, a solution to \eqref{eq: atkOptim} exists and is unique.} 

Throughout our analysis, we assume the defender can place investments on nodes preceding the target node to protect it, but not on the target itself (i.e. $x^*_g = 0$). This is a mild assumption, and resembles real-life scenarios where security investments on certain components cannot be made due to reasons such as conformance to standards, functional requirements, or ownership.  
\rev{Additionally, we assume that the defender has access to a \emph{sufficient} budget $B$, formally stated below. The proof is provided in Appendix \ref{app: sufficient_proof}.}
\begin{lemma}[Sufficient budget]
\label{lemma: sufficient-budget}
    \rev{For any given instance of the Stackelberg game on the attack graph, there exists a \emph{sufficient budget} $B$ such that if the optimal investment on a non-entry node $v_j$ is $x^*_j=0$ at the solution of \eqref{eq: atkOptim} given budget $B$, then $x^*_j=0$ at the solution of \eqref{eq: atkOptim} under any $B'\geq B$.}
\end{lemma} 

This choice still allows us to evaluate how the defender prioritizes the expenditure of a limited budget, without considering cases in which some nodes are not attended to due to lack of resources. \rev{In other words, if a node receives zero investment in the optimal profile under sufficient budget, it will continue receiving no investment even if the defender procures more security budget.} \rev{All following analysis is performed assuming sufficient budget is available.}\footnote{\rev{We discuss the insufficient budget case in Appendix \ref{app:insufficient}}.}
\section{Attack Graph Reductions}
\label{sec: reductions}

\begin{figure}[t]
    \centering
    \includegraphics[scale=0.2]{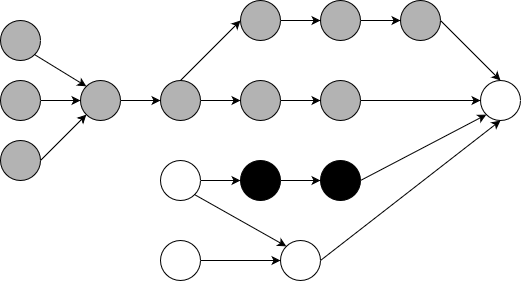}
    \caption{An attack graph in its original form.}
    \label{fig: original graph}
\vspace{0.1in}
    \captionsetup{justification=centering}
    \includegraphics[width=0.2\columnwidth]{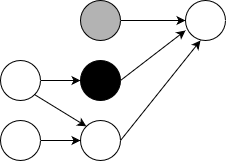}
    \caption{Reduced attack graph of Fig.~\ref{fig: original graph}, \rev{with the gray/black nodes in Fig.~\ref{fig: original graph} reduced to a single gray/black node.}}
    \label{fig: reduced graph}
\end{figure}

We begin our analysis by showing that the attack graph $\mathcal{G}$ of the security game described in Section~\ref{sec: framework} can be transformed into an ``equivalent'' reduced form which considerably simplifies the optimization problem in \eqref{eq: atkOptim}. Formally, we define equivalence between attack graphs as follows. 
\begin{definition}[Equivalent graphs]\label{def: equivalent}
Two attack graphs $\mathcal{G}_1$ and $\mathcal{G}_2$ are \emph{equivalent} if they have the same expected loss $\mathbf{L}^*$ under their respective optimal strategies $\mathbf{x}^*_1$ and $\mathbf{x}^*_2$. We denote this by $\mathcal{G}_1 \equiv \mathcal{G}_2$. 
\end{definition}

Our motivation for proposing such attack graph reductions is two-fold. First, our reduction procedure leads to an equivalent attack graph with a reduced number of nodes, which simplifies problem \eqref{eq: atkOptim} by reducing the number of decision variables. Moreover, we propose using this reduction algorithm in conjunction with our network re-design interventions presented in Section~\ref{sec: formations}. Specifically, we are in general interested in evaluating whether a network re-design intervention can be effective by reducing the expected loss in the network. To this end, it is sufficient to compare the losses on the reduced forms of the attack graphs before and after the intervention. This will in turn reduce the computational requirement when assessing different candidate interventions. \rev{An illustration of our reduction procedure's outcome is shown in Figs.~\ref{fig: original graph} and \ref{fig: reduced graph}.}
 
\rev{To see why the computational load of security assessment can be lowered by our approach, first} note that {through the addition of a variable, problem \eqref{eq: atkOptim} can be converted into a minimization problem with $|\mathcal{P}_{sg}|$ inequality constraints where $\mathcal{P}_{sg}$ denotes the set of all the attack paths leading to $v_g$. The current optimization problem has a total of $m = n + p + 1$ constraints, were $n$ indicates the number of assets, $p$ indicates the number of paths from the source nodes to the target, and the one additional constraint indicates the budget constraint. In general, using interior-point methods, solving this optimization problem with $n$ variables and $m$ constraints has an overall computational complexity of $O(n^{1.5}m^3\log(\frac{1}{\epsilon}))$, where $\epsilon$ indicates the degree of accuracy within which the solution is obtained \cite{hochbaum2007complexity, polik2010interior}. Accordingly, a reduction in the number of variables (together with the resulting reduction in the number of attack paths) will reduce the computation time to solve (\ref{eq: atkOptim}).} We provide numerical examples in Section~\ref{sec: application}.

\rev{We present our reduction procedure as a number of subroutines to be applied to series paths, parallel paths, and input nodes in the attack graph. All proofs are presented in Appendix~V and show that the attained reduced graph following each subroutine leads to an equivalent graph (in the sense of Definition~\ref{def: equivalent}) to the original graph.}

\subsection{Series path reductions}
We first present attack graph reductions which ultimately replace any series path $\{v_i \rightarrow v_{i+1} \rightarrow \cdots \rightarrow v_{i+n}\}$ with a single equivalent node. Formally, we say  $\{v_i \rightarrow v_{i+1} \rightarrow \cdots \rightarrow v_{i+n}\}$ is a series path if \rev{$Pre(v_{k+1})=\{v_{i}, v_{i+1}, \ldots v_{k}, v_{k+1}\}$ and $Post(v_{k}) = \{v_{k}, v_{k+1}, \ldots, v_{i+n}\}$ for all $k\in\{i, \ldots, i+n-1\}$}. For readability, the default loss $p_i^0$ for individual nodes is dropped in the remainder of this section. This is without loss of generality, as they can be subsumed in the stand-alone loss $L_i$ of the node.

{We begin by identifying nodes that will receive a zero investment at the optimal equilibrium profile, and show that these can be subsumed in their preceding nodes to obtain an equivalent attack graph.}

\begin{lemma}[Series zero investments]
\label{lemma:series_zero_investments}
{Consider a series link $\{v_i \rightarrow v_{i+1} \rightarrow \cdots v_{i+n-1} \rightarrow v_{i+n}\}$. 
\begin{itemize}
    \item Start at $m = i+n$. The pair of nodes $v_{m-1} \rightarrow v_m$ can be replaced with a node with $\kappa_{eq} = \kappa_{m-1}$ and $L_{eq} = L_{m-1} + L_{m}$, \rev{and $x_m^*=0$}, if and only if
    \begin{equation}
    \label{eq: series_end_conditions}
        \begin{gathered}
            \rev{L_{m-1} \geq (\frac{\kappa_{m}}{\kappa_{m-1}}-1)L_{m}~.}
        \end{gathered}
    \end{equation}
    \item If (\ref{eq: series_end_conditions}) is not satisfied, for $ i < m \leq i+n-1$, the pair of nodes $v_{m-1} \rightarrow v_m$ can be replaced with a node with $\kappa_{eq}=\kappa_{m-1}$ and $L_{eq}=L_{m-1} + L_{m}$, \rev{and $x_m^*=0$}, if and only if
    \begin{equation}
    \label{eq: series_mid_conditions}
        \begin{gathered}
            \kappa_{m-1} \geq \kappa_{m} \text{ or }     L_{m-1} \geq \Big(\frac{\frac{\kappa_{m-1}}{\kappa_{m-2}}-1}{1-\frac{\kappa_{m-1}}{\kappa_m}}\Big)L_m~.
        \end{gathered}
    \end{equation}
\end{itemize}
}

\end{lemma}

Intuitively, under the conditions in the lemma, either the earlier node $v_i$ provides higher marginal return on investment, or has a substantially higher stand-alone loss than the subsequent node $v_j$, and as such, the defender is better off adopting ``perimeter defense'' and investing all budget on the outer node $v_i$. 
Note that this finding is consistent with previous results {in \cite{abdallah2020behavioral, oruganti2021impact, abdallah2021morshed}}, which had studied the special case of $\kappa_i=\kappa_j$. Lemma~\ref{lemma:series_zero_investments} further extends these results as it identifies conditions under which the defender distributes her investments over inner nodes as well. 

Note that by {repeated application of} Lemma~\ref{lemma:series_zero_investments}, {working our way from the last node backwards to the first node}, we can convert any series path $\{v_i \rightarrow v_{i+1} \rightarrow \cdots \rightarrow v_{i+n}\}$ to a reduced form in which all remaining nodes should have non-zero investments at the optimal investment profile. Following this, we conduct the remaining series reduction as detailed below, which will result in any series path being replaced by a single equivalent node.

\begin{lemma}[Series reduction]
\label{lemma:series-reduction}
Consider a series path  $\{v_i \rightarrow v_{i+1} \rightarrow \cdots \rightarrow v_{i+n}\rightarrow v_t\}$. Assume that $x^*_j \neq 0$ at all of these nodes. Then, this path can be replaced by a single equivalent node $v_{eq}$ with $\kappa_{eq}=\kappa_i$ and stand-alone loss
    \begin{equation}
    \begin{split}
        L_{eq} = \frac{\kappa_{i+1}L_i}{\kappa_{i+1}-\kappa_i}\prod_{j = i+1}^{i+n}\Big(\frac{\kappa_{j-1}}{\kappa_{j+1}}\Big(\frac{\kappa_{j+1}-\kappa_{j}}{\kappa_{j}-\kappa_{j-1}}\Big)\frac{L_{j-1}}{L_{j}}\Big)^{\frac{-\kappa_i}{\kappa_{j}}} \\
        \Big(\frac{\kappa_{i+n}}{\kappa_{t}-\kappa_{i+n}}\Big(\frac{L_{i+n}}{L_t}\Big)\Big)
        ^{\frac{-\kappa_i}{\kappa_{i+n}}}~.
    \end{split}
    \label{eq:series-long}
    \end{equation}
\end{lemma}

\subsection{Parallel path reductions}
After performing the proposed series link reductions, the reduced attack graph can contain parallel paths of the form $\{(v_i \rightarrow v_{i+1} \rightarrow v_j), (v_i \rightarrow v_{i+2} \rightarrow v_j), \ldots, (v_i \rightarrow v_{i+n} \rightarrow v_{t})\}$, where $Pre(v_{i+j})=\{v_i\}$ and $Post(v_{i+j})=\{v_t\}$, for all $j \in I := \{1, 2, 3 \ldots, n\}$. In this section, we identify scenarios under which parallel paths of this form can be replaced by a single equivalent node $v_{eq}$.  

Similar to the series reduction case, we first identify cases in which we can determine, a priori, if one or more of the parallel nodes should receive zero investment under the optimal investment strategy, and can therefore remove them prior to solving for the optimal investment profile. 

\begin{lemma}[Parallel zero investments]\label{lemma: non-zero parallel}
{Consider a set of parallel paths $\{( v_i \rightarrow v_{i+1} \rightarrow v_g), ( v_i \rightarrow v_{i+2} \rightarrow v_g), \ldots, ( v_i \rightarrow v_{i+n} \rightarrow v_g)\}$ with $L_{i+1} < L_{i+2} < \cdots < L_{i+n}$. Let $\kappa_{par} := \sum_{k \in I} \frac{1}{\kappa_k}$. Then, \rev{$x^*_{i+1} = 0$} if and only if
\begin{equation}
\label{eq: parallel_loss_condition}
    L_i \geq \Big(\frac{1 - \kappa_i \kappa_{par}}{\kappa_i \kappa_{par}}\Big)(L_{i+1} + L_g)
\end{equation}}
\end{lemma}

Intuitively, the above lemma can be interpreted as follows. It may arise that due to the security attributes of the parallel nodes (which follow (\ref{eq: parallel_loss_condition})), the optimal action is to equate the losses across all paths to that of $\mathbf{L}_l = L_i + L_l + L_g$ where $v_l$ is the loss the parallel node with the least stand-alone loss $L_l$. 

{This lemma can be applied repeatedly: with $v_{i+1}$ receiving no investment, we can check Lemma \ref{lemma: non-zero parallel} over the remaining nodes until no additional reductions of this type are possible.} {Note also that in the special case when $\kappa_i = \kappa, \forall i$, through repeated application of Lemma~\ref{lemma: non-zero parallel}, the set of paths will be replaced with the single path containing the parallel node with the highest stand-alone loss; this matches our earlier results in the homogeneous $\kappa$ model  \cite{oruganti2021impact}.}

Following repeated application of Lemma~\ref{lemma: non-zero parallel}, all remaining parallel nodes in sets of the form $\{( v_i \rightarrow v_{i+1} \rightarrow v_g ), ( v_i \rightarrow v_{i+2} \rightarrow v_g ), \ldots, ( v_i \rightarrow v_{i+n} \rightarrow v_g )\}$ will receive non-zero investments. These sets can be further replaced by a single equivalent node, as shown in the following lemma. 

\begin{lemma}[Parallel reduction]
\label{lemma: parallel}
Consider a set of parallel paths, $\{( v_i \rightarrow v_{i+1} \rightarrow v_g ), ( v_i \rightarrow v_{i+2} \rightarrow v_g ), \ldots, ( v_i \rightarrow v_{i+n} \rightarrow v_g )\}$ with $\kappa_{par} = \sum_{r=i+1}^{i+n}\frac{1}{\kappa_r}$ such that the conditions of Lemma \ref{lemma: non-zero parallel} are not satisfied. Then $\mathcal{G}$ can be replaced with a single equivalent node $v_{eq}$ with $k_{eq} = k_i$ and
\begin{equation}
L_{eq} = \frac{L_i}{1-\kappa_i \kappa_{par}}\prod_{j=i+1}^{i+n}\Big(\frac{\kappa_i\kappa_{par}}{1-\kappa_i\kappa_{par}}\Big( \frac{L_i}{L_{j}+L_g}\Big)\Big)^{\frac{-\kappa_i}{\kappa_{j}}}
\end{equation}
\end{lemma}

It is to be noted that the outcome of these steps can lead to parallel paths being reduced to series links. In such a scenario, further reduction of the attack graph can be achieved by looping between the procedures in Lemmas~\ref{lemma:series_zero_investments}-\ref{lemma: parallel}, until no further series or parallel reduction is possible.
\subsection{Input node reductions}
\label{sec: additional input} 

Finally, we look at a possible reduction of multiple input nodes. Similar to the previous sections, we begin by providing a condition on input nodes given which one can a priori guarantee that their first successor node, $v_t$, will receives a zero investment in the optimal investment profile. 

\begin{lemma}\label{lemma:zero destination node} Consider a set of paths $\mathcal{G} = \{(v_{i} \rightarrow v_t \rightarrow v_g), (v_{i+1} \rightarrow v_{t} \rightarrow v_g), \cdots, (v_{i+n} \rightarrow v_t \rightarrow v_g)\}$ with $L_j = L, ~\forall j \in \mathcal{I} = \{i, i+1, i+2, \ldots, i+n\}$. Let $\kappa_{par} = \sum_{r=i}^{i+n}\frac{1}{\kappa_r}$.  Then, \rev{$x^*_t = 0$} if and only if $\kappa_t\kappa_{par} \leq 1$ or $L \geq ( 1-\kappa_t\kappa_{par})(L_t+L_g)$. 
\end{lemma}

Similar to the result in \cite{oruganti2021impact}, this lemma states that the defender is better off choosing a ``perimeter defense'' if the stand-alone loss $v_t$ is substantially lower than the input nodes; it further extends that result by showing that the same is true if the inner node $v_t$ has a relatively lower return-on-investment. 

We now look at the case when the conditions of Lemma \ref{lemma:zero destination node} are not satisfied (i.e., $x^*_t \neq 0$), and show that multiple input nodes can be replaced by a single equivalent node. 

\begin{lemma}[Source node reduction]
\label{lemma: inputs}
Consider a set of input paths $\{(v_{i} \rightarrow v_{t}), (v_{i+1} \rightarrow v_{t}), \ldots, (v_{i+n} \rightarrow v_t) \}$ with equal stand-alone loss input nodes, i.e., $L_{i+1} = \cdots = L_{i+n} = L$. Assume the conditions of Lemma~\ref{lemma:zero destination node} are not met. Then, this set can be replaced with a single equivalent node $v_{eq}$ such that
\begin{equation}
    L_{eq} = \frac{L\kappa_i\kappa_{par}}{\kappa_i\kappa_{par}-1}\Big(\frac{1}{\kappa_i\kappa_{par}-1}\Big(\frac{L}{L_t}\Big)\Big)^{\frac{-1}{\kappa_i\kappa_{par}}}
\end{equation}
where $k_{par} = \sum_{r=i}^{i+n}\frac{1}{k_{r}}$ and $k_{eq} = \frac{1}{k_{par}}$.
\end{lemma}

\subsection{Reduction algorithm}

\rev{We now present our proposed attack graph reduction procedure in Algorithm~\ref{alg: reduction}. The statement and proof of Proposition~\ref{prop:reduction-alg} are based on the sequence of Lemmas \ref{lemma:series_zero_investments}-\ref{lemma: inputs} presented earlier. This proposition generalizes our earlier work \cite{oruganti2021impact} as well as related results in prior works \cite{abdallah2020behavioral, abdallah2021morshed}.}

\begin{proposition}\label{prop:reduction-alg}
Given a sufficient budget $B$, Algorithm~\ref{alg: reduction} leads to an equivalent reduced form $\mathcal{G}_r$ of attack graph $\mathcal{G}$. 
\end{proposition}

\begin{algorithm}[!h]
\caption{Reduction of attack graph $\mathcal{G}$ to an equivalent  $\mathcal{G}_r$} 
\label{alg: reduction}
		\begin{algorithmic}[1]
			\State Input: An attack graph $\mathcal{G}$
			\State Output: An equivalent reduced attack graph $\mathcal{G}_r$
    			\While {series or parallel paths reducible}
    			\State \textbf{\textit{Series Paths Reduction Step}:}
    			\State Gather all series paths in graph $\mathcal{G}$
    			\State Lemma~\ref{lemma:series_zero_investments}: remove series nodes with no investment
    			\State Lemma~\ref{lemma:series-reduction}: reduce remaining series links to one node
    			\State Update the set of series paths in $\mathcal{G}$ accordingly
    			\State \textbf{\textit{Parallel Paths Reduction Step}:}
    			\State Gather all parallel paths in graph $\mathcal{G}$
    			\State Lemma~\ref{lemma: non-zero parallel}: remove (some) parallel paths
    			\State Lemma~\ref{lemma: parallel}: reduce (some) parallel paths to one node
    			\State Update the set of parallel paths in $\mathcal{G}$ accordingly
    			\EndWhile
    			\State \textbf{\textit{Input Node Reduction Step}:}
    			\State Consider all input nodes in graph $\mathcal{G}$
    			\State Apply Lemma \ref{lemma:zero destination node} if possible, else Lemma \ref{lemma: inputs}, to remove (some) of the input nodes
    			\State Update the set of input nodes in $\mathcal{G}$ accordingly
			\State \Return reduced graph $\mathcal{G}_r$
		\end{algorithmic}
\end{algorithm}

\rev{In Appendix~\ref{app: optimal_investments}, we further detail how the optimal investments obtained from the reduced graph $\mathcal{G}_r$ can be mapped back to the optimal investments on the original graph $\mathcal{G}$.}

\section{Network Design Interventions}
\label{sec: formations}

While the attacker-defender games of Section~\ref{sec: framework} have been studied in a number of prior works in the homogeneous return-on-investment case {(e.g., \cite{abdallah2020behavioral, hota2018game})}, their focus, similar to the analysis presented in Section~\ref{sec: reductions}, has been on the study of the optimal investment strategy \emph{given} a fixed network. In addition to extending these models by considering heterogeneous return-on-investments $\kappa$, this paper further evaluates the use of an orthogonal set of defender actions, in the form of \emph{network design interventions}. 

To illustrate the main ideas, we consider a minimal base network and four re-design actions as illustrated in Fig.~\ref{fig: network redesign}: (a) adding a node in series; (b) adding a node in parallel; (c) a combination of series and parallel additions; and (d) adding an new entry node. We compare the overall loss on the networks obtained through these actions against those of the base network, and provide (intuitive) interpretations for the potential effects of each type of intervention. 

These four types of interventions can be made in any general CPS. As mentioned earlier, our reduction approach in Section~\ref{sec: reductions} can be used to simplify the task of comparing the expected losses following these interventions. {In Section~\ref{sec: application}, we will elaborate on the effect of these interventions in more general networks, {and show that they match the intuitions obtained from the analysis of the base network in this section}, using numerical examples motivated by applications in industrial cyber-physical systems.}  

\begin{figure}[t]
    \centering
    \includegraphics[scale=0.2]{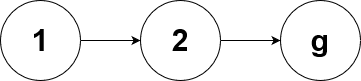}
    \caption{A minimal base network.}
    \label{fig: base}
    \vspace{0.1in}
    \captionsetup{justification=centering}
    \includegraphics[width=0.7\columnwidth]{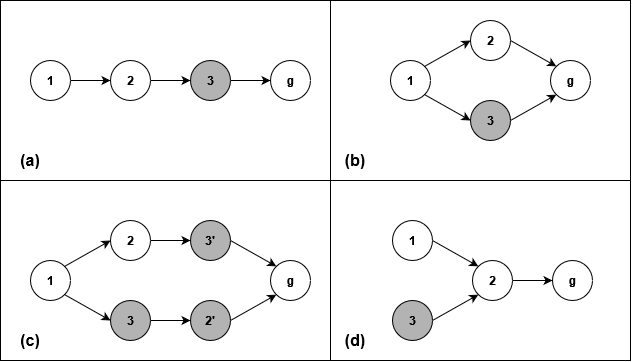}
    \caption{Network design intervention alternatives.}
    \label{fig: network redesign}
\end{figure}

\subsubsection*{The Base network} Consider the minimal attack graph  
shown in Fig.~\ref{fig: base}. This attack graph is minimal in the sense that the target node $v_g$ is an interior node of the network, accessible only through a stepping-stone attack by compromising the entry node ($v_1$) as well as an intermediate node ($v_2$). 

For simplicity, we let $p_i^0=p,~ \forall i$ for the analytical results in this section; in the numerical illustrations, we additionally set $p=1$ and highlight the impact of other problem parameters. We also assume the problem parameters are such that the conditions of Lemma \ref{lemma:series_zero_investments} are not met, so that both nodes $v_1$ and $v_2$ receive non-zero investments at equilibrium. The expected loss for the base network in this case is:
\begin{equation*}
\label{eq: base 2}
    \mathbf{L}^*_{\text{b}} = \frac{L_1\kappa_2p}{\kappa_2-\kappa_1}\Big(\frac{\kappa_1}{\kappa_2-\kappa_1}\Big(\frac{L_1p}{L_2p^2+L_gp^3}\Big) \Big)^{\frac{-\kappa_1}{\kappa_2}}e^{-k_1B} 
\end{equation*}
We will next assess which design interventions can help lower this expected loss. 

\subsubsection{Series connection: increased endurance}
\label{subsec: series}
The first intervention we consider is that of an addition of a node in series, as illustrated in Fig.~\ref{fig: network redesign}(a). Security interventions in the form of adding encryption devices or requiring (stronger) passwords can be represented as this type of network re-design. Intuitively, we might expect that the addition of a series node will increase the endurance of the system as the attacker now has to compromise an extra node to get to the target.

We consider two possibilities for such interventions: closer to the target node (strengthening the core of the network) vs. closer to the entry node (strengthening perimeter defenses). {In the former case when $v_3$ is added after $v_2$, if the conditions of Lemma \ref{lemma:series_zero_investments} apply on $v_3$, then $x^*_3=0$. The loss of the modified network then will be similar to $\mathbf{L}^*_{\text{b}}$, but with $L_2$ replaced by $L_2+L_3$ {and $L_g$ multiplied with $p^4$}. As $\mathbf{L}^*_{\text{b}}$ is increasing in $L_2$ (since $\frac{\partial \mathbf{L}^*_{\text{b}}}{\partial L_2} > 0$) this means that {while the series addition does reduce the probability of attack on the downstream node ($v_g$ here), there may arise scenarios were this may not offset the increase in loss of the node immediately upstream ($v_2$ here).} This means that, perhaps counter intuitively, attempts at ``strengthening'' the core of the network with components with 
lower return-on-investment or {higher} safety criticality ({higher} $L$) tends to backfire and {increase} the total loss. 

Next, we look at the case when the conditions of Lemma \ref{lemma:series_zero_investments} for $v_3$ are not met, which implies $x^*_3 \neq 0$. {Depending on whether $v_3$ is introduced before ($pre$) or after ($post$) $v_2$}, the equilibrium expected losses obtained are:} 
\begin{multline*}
    \mathbf{L}^*_{\text{srs, post}} = \frac{L_1\kappa_2p}{\kappa_2-\kappa_1}
    \Big(\frac{\kappa_1}{\kappa_3}\Big(\frac{\kappa_3-\kappa_2}{\kappa_2-\kappa_1}\Big)\frac{L_1p}{L_2p^2}\Big)^{\frac{-\kappa_1}{\kappa_2}} \\
    \Big(\Big(\frac{\kappa_2}{\kappa_3-\kappa_2}\Big)\frac{L_2p^2}{L_3p^3+L_gp^4}\Big)^{\frac{-\kappa_1}{\kappa_3}}e^{-k_1B}
\end{multline*}

\begin{multline*}
    \mathbf{L}^*_{\text{srs, pre}} = \frac{L_1\kappa_3p}{\kappa_3-\kappa_1}
    \Big(\frac{\kappa_1}{\kappa_2}\Big(\frac{\kappa_2-\kappa_3}{\kappa_3-\kappa_1}\Big)\frac{L_1p}{L_3p^2}\Big)^{\frac{-\kappa_1}{\kappa_3}} \\
    \Big(\Big(\frac{\kappa_3}{\kappa_2-\kappa_3}\Big)\frac{L_3p^2}{L_2p^3+L_gp^4}\Big)^{\frac{-\kappa_1}{\kappa_2}}e^{-k_1B}
\end{multline*}

We compare these expected losses against that of the base network numerically. \rev{In Fig.~\ref{fig: series_additional} we fix the values of $L_1, L_2$, and $L_g$ and 
compare the resulting loss $L^*_b$ (indicated using the red curve) against the losses $\mathbf{L}^*_{\text{srs, pre}}$ and $\mathbf{L}^*_{\text{srs, post}}$  
as a function of $L_3$ and $\kappa_3$ of the added node}. {First, we note that in both cases, the total expected loss increases with $L_3$, meaning that the added node should itself have a low stand-alone loss for the expected loss to decrease relative to the base network.} Further, it can be seen that for a $v_3$ added downstream (closer to the target) to lower the expected loss relative to the base network, it has to be a node with a relatively \emph{high} return-on-investment. Adding the same node upstream would lead to a decrease in expected loss at lower $\kappa_3$. 

\rev{\textit{Impacts of interventions on expected loss vs. attack probability on individual nodes:} We note that there may exist a trade-off between minimizing total expected loss and the probability of attack on a given node. Figure \ref{fig: pareto} illustrates a scenario where the attack probability on $v_2$ (given by $p^*_2 = p^0_1p^0_2e^{-\kappa_1x^*_1-\kappa_2x^*_2}$) changes depending on the security attributes of a node $v_3$ added immediately downstream. It can be seen that as $\kappa_3$ increases, higher security investments on $v_3$ lead to lower total expected losses, while the reduced investment on $v_2$ leads to a higher probability of attack on that node. That is, the designer opts to make interventions that decrease overall loss, despite the (negative) impacts it may have on some of the individual assets.}

{\textbf{In summary}, we observe that the addition of a series node $v_3$ can help strengthen the network if the added node has sufficiently low stand-alone loss $L_3$ and sufficiently high return-on-investment $\kappa_3$, with the benefits being higher if it is feasible to add the node closer to the perimeter of the network.} 

\begin{figure}
\centering
\includegraphics[width=\linewidth]{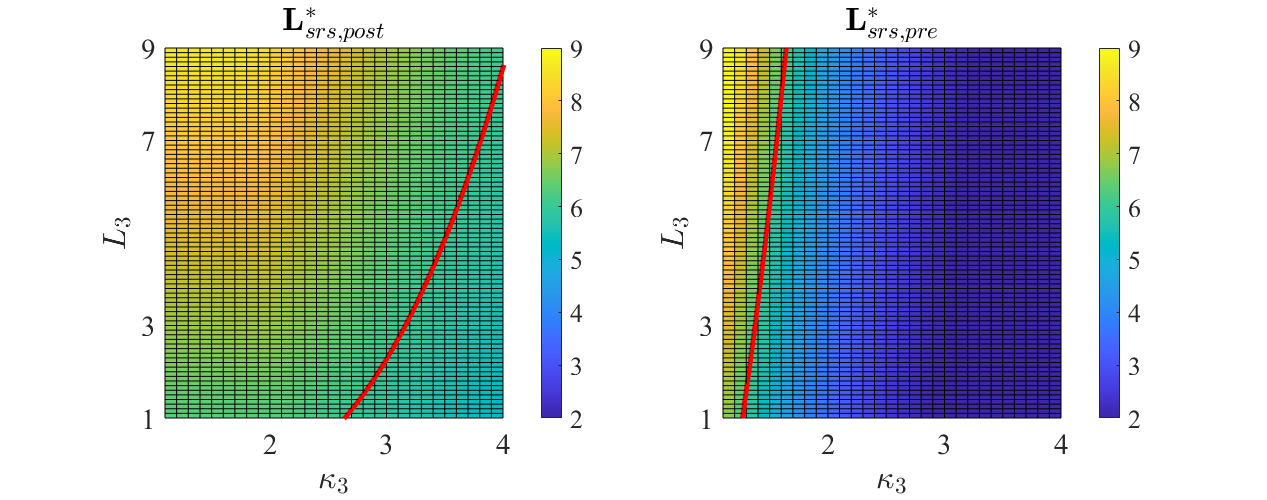}
  \caption{Adding $v_3$ after (left) and before (right) $v_2$ in the base network. \rev{The red line indicates the level curve of the base expected loss $\mathbf{L}^*_b = 6$. All values to the left of the curve indicate a higher loss compared to the base.}}
  \label{fig: series_additional}
\end{figure}

\begin{figure*}[]
\centering
\begin{subfigure}{0.5\columnwidth}
    \centering
    \includegraphics[scale=0.3]{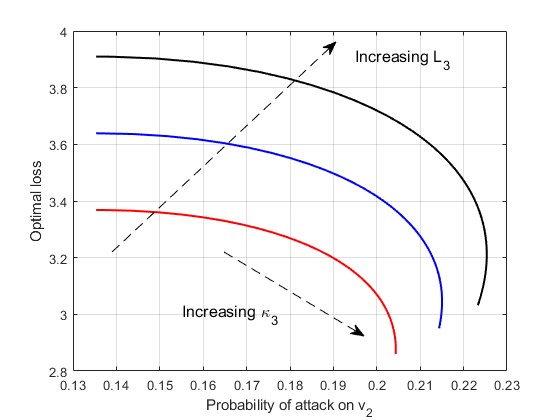}
    \caption{\rev{Trade-off between $p^*_2$ and $\mathbf{L}^*_{\textrm{srs,post}}$ after a series intervention}}
    \label{fig: pareto}
\end{subfigure}
\hfill
\begin{subfigure}{0.5\columnwidth}
    \centering
    \includegraphics[scale=0.3]{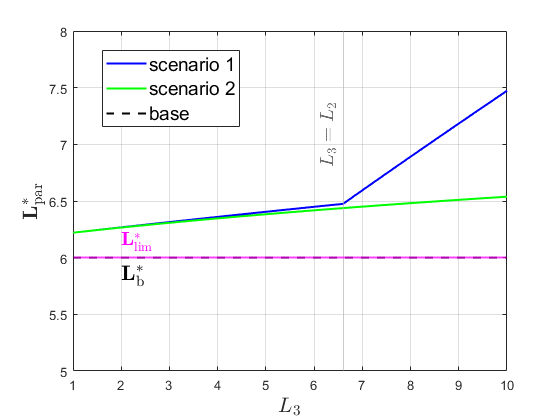}
    \caption{Expected loss vs. $L_3$; parallel intervention (Fig.~\ref{fig: network redesign}(b))}
    \label{fig: parallel 4 with 2}
\end{subfigure}
\hfill
\begin{subfigure}{0.5\columnwidth}
    \centering
    \includegraphics[scale=0.3]{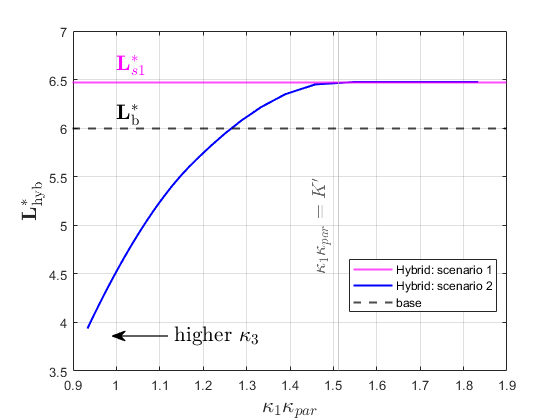}
    \caption{Expected loss vs. $\kappa_3$; hybrid intervention (Fig.~\ref{fig: network redesign}(c))}
    \label{fig: par hyb}
\end{subfigure}
\caption{Expected loss in different numerical scenarios studied}
\label{fig: test}
\end{figure*}

\subsubsection{Parallel connection: structural/physical redundancy}
\label{subsec: parallel}
Next we study the addition of a parallel connection, as illustrated in Fig. \ref{fig: network redesign}(b). This intervention can be seen as improving the number of redundant components in the system to improve its tolerance to physical failures; examples include, adding redundant communication lines, back-up generators, etc. 

We split our analysis into two scenarios: 1) when the condition of Lemma \ref{lemma: non-zero parallel} is satisfied, and 2) when it is not. In scenario 1, the node with the lower stand-alone loss between $v_2$ and $v_3$ receives no investment. The total expected loss here depends on the relative value of $L_2$ and $L_3$: 
\[ \mathbf{L}^*_{\text{par}} = \begin{cases}
        (L_1p +L_3p^2 +L_gp^3)\Big(\frac{L_3 +L_gp}{L_2+L_gp}\Big)^{\frac{-\kappa_1}{\kappa_2}}e^{-\kappa_1B}, \\
        \text{if } L_2 \geq L_3; 
        & \\
        (L_1p +L_2p^2 +L_gp^3)\Big(\frac{L_2 +L_gp}{L_3+L_gp}\Big)^{\frac{-\kappa_1}{\kappa_3}}e^{-\kappa_1B}, \\ \text{if } L_2 < L_3.
        \end{cases} \]        
In this scenario, since $\kappa_2 > \kappa_1$ (assumption from the base case), even for low $L_3$, $\mathbf{L}^*_{\text{par}} > \mathbf{L}^*_{\text{b}}$. Additionally, since $\frac{\partial\mathbf{L}^*_{\text{par}}}{\partial L_3} > 0$, the loss increases with increasing $L_3$. Increasing $\kappa_3$ (until the condition of Lemma \ref{lemma: non-zero parallel} satisfied) has no effect when $L_2 \geq L_3$, and a limited effect in reducing loss when $L_2<L_3$. This means that overall, the addition of $v_3$ in this scenario tends to increase the expected loss relative to the base case.  

In scenario 2, letting $\kappa_{par} = \frac{1}{\kappa_2} + \frac{1}{\kappa_3}$:
\begin{multline*}
    \mathbf{L}^*_{\text{par}} = \frac{L_1p}{1-\kappa_1\kappa_{par}}\Big(\frac{\kappa_1\kappa_{par}}{1-\kappa_1\kappa_{par}}\Big(\frac{L_1p}{L_2p^2 + L_gp^3}\Big)\Big)^{\frac{-\kappa_1}{\kappa_2}}\\
\Big(\frac{\kappa_1\kappa_{par}}{1-\kappa_1\kappa_{par}}\Big(\frac{L_1p}{L_3p^2 + L_gp^3}\Big)\Big)^{\frac{-\kappa_1}{\kappa_3}}e^{-\kappa_1B}.
\end{multline*}

Figure \ref{fig: parallel 4 with 2} numerically illustrates the effect of varying $L_3$ on $\mathbf{L}^*_{par}$ in both scenarios, with the other nodes fixed at the same attributes as in Fig.~\ref{fig: series_additional}, and $\kappa_3$ is chosen such that it is the lowest value at which both $v_2$ and $v_3$ receive investment in scenario 2. Firstly, we observe that the parallel case has a higher total expected loss than the base, irrespective of the scenario. This is to be expected (in general) as the number of paths to the target have increased and the designer has to potentially split investments over multiple paths. Similar to the series case, it can be shown that the total expected loss keeps increasing with increasing $L_3$ and decreases with increasing $\kappa_3$. As $\kappa_3$ increases, the investment on $v_3$ decreases as it gives more return-on-investment. In the limiting case (very high $\kappa_3$), $v_3$ receives very little investment with all the budget going on $v_2$ to balance losses across both paths. Hence the total expected loss in the limiting case would be $\mathbf{L}^*_{{\text{limit}}} = L_1p + (L_2p^2 + L_gp^3)e^{-\kappa_2 B}$, so that $\mathbf{L}^*_{{\text{limit}}} = \mathbf{L}^*_{{\text{b}}}$.

\textbf{In summary}, adding a structurally redundant node $v_3$ does not in general improve the security posture compared to the base architecture. At best, with (very) high return-on-investment $\kappa_3$ or low stand-alone loss $L_3$, the total expected loss remains close to the base case. Overall, adding such redundancies can improve operational reliability, but increases the attack surface (and hence expected loss) in the CPS.

\subsubsection{Hybrid connection: functional/informational redundancy}
\label{subsec: hybrid}

Our next intervention assesses the impact of introducing {functional redundancy} in a system. A functionally redundant component can be used to carry out the same tasks as an existing node; for instance an additional sensor can be added to attain signals for health monitoring, or anomaly detection and isolation. While functioning independently, information from such components is generally used in unison for decision making. As a result, a successful attacker would need to compromise \emph{both} nodes (as least to some extent) to proceed in the stepping stone attack towards $v_g$. To capture this, we add the nodes $v_{2'}$ and $v_{3'}$ after $v_3$ and $v_2$, respectively. 

Considering the similarity in the type of nodes $v_2$ and $v_3$, we assume $L_2 = L_3 = L$, and let the auxiliary nodes have zero loss, i.e., $L_{2'}=L_{3'}=0$. We again denote $\kappa_{par} = \frac{1}{\kappa_2} + \frac{1}{\kappa_3}$. The defender equalizes the expected losses over both paths, leading to $x^*_2 \kappa_2 = x^*_3 \kappa_3$. Here, $x^*_2 \neq 0$ and $x^*_3 \neq 0$ when
{\begin{gather*}
    L_1 < \Big(\frac{1}{\kappa_1\kappa_{par}} - 1\Big)Lp + \Big( \frac{2}{\kappa_1\kappa_{par}} - 1\Big)L_gp^3~.
\end{gather*}
Similar to the other sections, we divide our analysis into two scenarios: 1) when $\kappa_1\kappa_{par} > 2$ and 2) when $\kappa_1\kappa_{par} < 2$. In scenario 1, it can be seen that the above condition fails irrespective of the other security attributes leading to $x^*_2 = x^*_3 = 0$. Similar to the previous sections, if the stand-alone loss of the root of the paths is significantly more than that of the other nodes, all investment is made on the root. Since all of the parallel nodes have the same stand-alone loss, all of them receive zero investment. In this case, we get the total expected loss to be $\mathbf{L}^*_{\text{s1}} =  (L_1p +Lp^2 +L_gp^4)e^{-\kappa_1B}$ which does not depend on $\kappa_3$. Note that unlike the parallel case, $L_g$ is multiplied with $p^4$ instead of $p^3$. This reflects the additional step that the attacker must perform to compromise the target. {As a result, the addition of such functional redundancies (hybrid nodes) can \emph{reduce} loss compared to the base network.} 

We next illustrate the effect of varying $\kappa_3$ numerically, in both scenarios, in Figure \ref{fig: par hyb}. Decreasing $\kappa_1\kappa_{par}$ (increasing $\kappa_3$) beyond $K' = \frac{Lp + 2L_gp^3}{L_1 + Lp + L_gp^3}$ leads to decreasing $\mathbf{L}^*_{\text{hyb}}$. This means that unlike the parallel case, the hybrid architecture allows a significant reduction in total expected loss at higher $\kappa_3$. This is expected: although new paths to the target are included and it may seem that the the attack surface has increased, each path is more robust and harder to compromise, since the information from all paths is fused at the target. 

\textbf{In summary,} the addition of functionally/informationally redundant nodes can \rev{decrease} the expected loss. This is because the probability of compromise \emph{at} the target node is reduced with the addition of the hybrid node due to the additional series components included in these nodes. 
Further improvements can be achieved if the return-on-investment of the additional node $v_3$ also has high return-on-investment $\kappa_3$.}

\subsubsection{Additional input node: new features}
\label{subsec: input} 
We finally look at the effect of adding entry nodes to the network, illustrated in Figure \ref{fig: network redesign}(d). This intervention represents scenarios involving adding additional features or functionalities; e.g., adding Bluetooth, wireless connectivity, etc., to improve user experience. \rev{Considering equal stand-alone losses $L$ for the input nodes}, we first consider the case when the conditions of Lemma \ref{lemma:zero destination node} are satisfied. {Informally, this happens when the return-on-investment or the stand-alone loss of the entry nodes is sufficiently higher than the next node $v_2$ and hence $x^*_2 = 0$. In this case, the total expected loss is
\[\mathbf{L}^*_{\text{inp}} = (Lp + L_2p^2 + L_gp^3)e^{\frac{-B}{\kappa_{inp}}}\]
where $\kappa_{inp} = \frac{1}{\kappa_1} + \frac{1}{\kappa_3}$. When the conditions of Lemma \ref{lemma:zero destination node} are not met, the total expected loss is
\[ \mathbf{L}^*_{\text{inp}} = \\ \tfrac{L\kappa_2\kappa_{inp}}{\kappa_2\kappa_{inp}-1}\Big(\tfrac{1}{\kappa_2\kappa_{inp}-1}\Big(\tfrac{Lp}{L_2p^2+L_gp^3}\Big)\Big)^{\frac{-1}{\kappa_2\kappa_{inp}}}e^{\frac{-B}{\kappa_{inp}}}.\]
It can be shown (analytically and numerically) that both effective losses are higher compared to the base case. This follows intuition, since adding more functionalities only leads to a larger attack surface, without providing any downstream benefits as all prior attack paths remain unaffected.}

\section{Applications}
\label{sec: application}

In this section, we provide numerical experiments to illustrate our previous analysis in two applications: a remote attack on an industrial SCADA system, where the goal of the attacker is to maliciously control physical actuators, and an attack on an automotive system, where the goal of the attacker is to remotely access the Controller Area Network (CAN) to send malicious commands.  All computations are done using the CasADi optimization toolbox \cite{andersson2019casadi} with the IPOPT solver on a generic laptop using an Intel i7 CPU @ 2.8 GHz.

\subsection{Remote attack on a SCADA system}
We look at an attack on an industrial CPS discussed in \cite{few2021case}. Here the attacker tries to obtain the control of actuators by performing the following steps: 1) gaining administrative privileges; 2) bypassing DMZ firewalls; and 3) gaining access to an industrial PLC. A simplified version of the attack graph is shown in Fig.~\ref{fig: scada attack graph}. The descriptions and security attributes of each node are given in Table~\ref{table: scada description} \rev{in Appendix \ref{app:experiment-tables}}.
We utilize the attack scoring mechanisms, the Common Vulnerability Scoring System (CVSS) \cite{cvss} to quantify $p^0$. We quantify the respective stand-alone losses $L$ for each node subjectively based on relative critically to the safety of the system.\footnote{Note that security quantification generally involves significant expert judgement; the values here are chosen to showcase the proposed methodology.} 

\begin{figure}[t]
    \centering
    \includegraphics[scale=0.3]{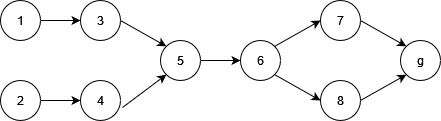}
    \caption{Attack graph for a remote attack on an industrial CPS}
    \label{fig: scada attack graph}
\end{figure}
 
We first solve \eqref{eq: atkOptim} for this attack graph. This problem has 9 variables, 12 inequality constraints, and 1 equality constraint, and it took around 0.030s to compute the optimal solution. With a budget of $B=5$ units, the optimal investment strategy is $\mathbf{x}^*_{\text{scada}} = \{1.4689, 1.4689, 0, 0, 2.0447, 0, 0, 0.0174\}$ and the expected loss is $\mathbf{L}^*_{\text{scada}}=586.67$.

Next, we apply our reduction algorithm from Proposition \ref{prop:reduction-alg} to first reduce the attack graph, and then find the optimal solution. The reduction steps are detailed below: 
\begin{itemize}
    \item First, we identify the two series links $\{v_1 \rightarrow v_3\}$ and $\{v_2 \rightarrow v_4\}$. We observe that $\kappa_3 = \kappa_1$ and $\kappa_4 = \kappa_2$. From Lemma \ref{lemma:series_zero_investments}, we immediately obtain $x^*_3 = x^*_4 = 0$. Nodes $\{v_1 \rightarrow v_3\}$ and $\{v_2 \rightarrow v_4\}$ can be replaced with equivalent nodes $v_{13}$ and $v_{24}$. 
    \item Next, we identify the parallel links $\{v_6 \rightarrow v_7 \rightarrow v_g\}$ and $\{v_6 \rightarrow v_8 \rightarrow v_g\}$. Since $\kappa_6\kappa_{78} > 1$, where $\kappa_{78} = \frac{1}{\kappa_7} + \frac{1}{\kappa_8}$, from Lemma \ref{lemma: non-zero parallel} we obtain $x^*_7 = 0$. 
    \item Next we identify the series link $\{v_5 \rightarrow v_6 \rightarrow v_8\}$ and observe that $\kappa_5 = \kappa_6$. From Lemma \ref{lemma:series_zero_investments}, we get $x^*_6 = 0$ followed by a reduction of this link to a single node $v_{56}$.
    \item We finally look at the input nodes $\{v_{13} \rightarrow v_5\}$ and $\{v_{24} \rightarrow v_5\}$. Let $L_{13} = L_{24} = L_{12}$ and $\kappa_{12} = \frac{1}{\kappa_1} + \frac{1}{\kappa_2}$. We observe that both the conditions from Lemma \ref{lemma:zero destination node} are not satisfied, i.e. $\kappa_5 \kappa_{12} > 1 $ and $L_{12} > (\kappa_5 \kappa_{12} - 1)L_{56}$. From Lemma \ref{lemma:zero destination node} we get $x^*_5 \neq 0$ and from Lemma \ref{lemma: inputs}, nodes $\{(v_{13}  \rightarrow v_{56}), (v_{24} \rightarrow v_{56})\}$ can be replaced with a single equivalent node $v_{in}$.
\end{itemize}

Using this procedure, the attack graph in Fig. \ref{fig: scada attack graph} is reduced to $v_{in} \rightarrow v_8 \rightarrow v_g$. The optimal loss on this network can again be found to be $\mathbf{L}^*_{\text{scada}}=586.67$. 
This optimization problem now has only 3 variables, 3 inequality constraints, and 1 equality constraint, and took around 0.02s to solve. This is a 50\% reduction in the computation time. While the absolute speed improvement seems minimal in this small problem, for much larger systems (for example, large-scale electric grids), we will get a significant improvement in absolute runtime as well. 

Additionally, we look at how our obtained optimal defense strategy compares with a baseline ``perimeter defense'' strategy. Such a perimeter defense strategy is similar to the ones recommended by \cite{oruganti2021impact} and \cite{abdallah2020behavioral} (where such strategies are shown to be optimal in the homogeneous $\kappa$ case and with rational defenders). Under such a strategy, only input nodes $v_1$ and $v_2$ would receive an investment of $2.25$ each. The total expected loss in this case would be $2.09\times 10^4$, which is much larger than $\mathbf{L}^*$. This highlights how accounting for the heterogeneity of assets in their return-on-investment can substantially impact optimal defense strategies. 

\subsection{Remote attack on an automotive system}

\subsubsection{System overview and quantification}
Next, we consider the remote attack on an automotive Controller Area Network (CAN), similar to the ones reported in \cite{miller2015remote} and \cite{cai20190}. Modern automotive systems provide many connectivity features such as Wi-Fi, Bluetooth (BT), cellular connectivity (CELL), and physical ports such as USB ports on the infotainment module. OBD-II ports (legally mandated for vehicle diagnostics) allow for a physical connection to the vehicle. In the worst-case, vulnerabilities in the communication protocols, firmware of the Electronic Control Units (ECUs), and software vulnerabilities on the infotainment module may allow access to the CAN and control of certain safety-critical actuators. 

Once notified of a vulnerability or a flaw, the manufacturer is faced with making a decision on (optimal) resource allocation for immediate as well as long-term security of the vehicle \cite{cai20190}. This decision is complicated further by: 1) the number of vulnerabilities with varying impacts and complexity, and 2) fleet exposure based on hardware and software version combinations. We apply our proposed approach on this problem to showcase its benefits and possible insights in such scenarios. The simplified attack graph from the attacks described in \cite{cai20190} is shown in Fig.~\ref{fig: automotive atk grf} and the descriptions of the nodes are provided in Table~\ref{table: autoAtkGrf description} \rev{in Appendix \ref{app:experiment-tables}}. 

\begin{figure}[t]
    \centering
    \includegraphics[scale=0.3]{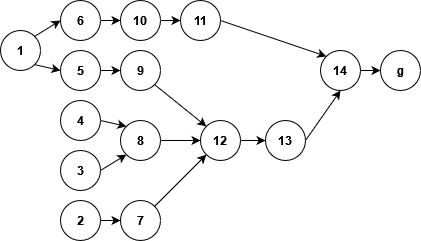}
    \caption{Attack graph for an 
    attack on an automotive system}
    \label{fig: automotive atk grf}
\end{figure}

We have set the numerical values of the nodes' attributes as follows. For simplicity, we classify the base probability of successful attack $p^0$ using a \{high, medium, low, very low\} ordinal scale with probabilities of 1, 0.75, 0.5, and 0.25, respectively. The remaning security attributes for each node are derived subjectively following discussions provided in \cite{cai20190}. For quantifying $\kappa$, we have taken into account a few factors such as user interaction with these units, fleet exposure, and individual criticality. For example, {Node-10} which represents the firmware on the telematics unit itself, has been given a low $\kappa$ as it may require user interaction, and as such we assume that significant security improvement on this unit is harder as it may degrade customer experience. {Node-6} on the other hand, represents remote communication procedures between back-end manufacturer servers and the telematics unit to perform diagnostic operations without user interaction, and hence we have set it as having a higher $\kappa$. 

\subsubsection{Optimal investment decisions}  
For this system, due to standardization and being external to the manufacturer, we assume that no investment is possible on Node-1 ($x^*_1 = 0$). Assuming a budget $B = 5$ units, we obtain  $\mathbf{x}^*_{\text{veh}} = [0,   0,   0,   0,    0.8179,    2.8317,    0.0956,    0.3820,    0,  0.5796,  0,    0.2932,  \\ 0,    0]$ and $\mathbf{L}^*_{\text{veh}} = 1.8837$. 

Since the numbers depend heavily on the quantification of the security attributes, we only interpret the \emph{ranking} of the investments. It can be seen that Nodes 5, 6, and 10 receive the highest investment. This indicates that immediate resources, whether that be time, number of personnel, or monetary resources, should be used for securing the nodes responsible for remote access; in fact, according to \cite{cai20190}, the manufacturer prioritized Nodes 6 and 10 in its defense strategy.  

  \subsubsection{Design interventions} {When the product (nodes) are in post-production, direct investment on these nodes may not be possible, and the manufacturer may instead perform a network redesign. The chosen countermeasure by the manufacturer in reality was to forward the remote request from the backend server (received from the customer) to another server which updates the configuration on the vehicle \emph{only} by HTTPS allowing for a higher payload and security \cite{cai20190}. To perform the intended command, the requested configuration must match with the updated configuration. This can be seen as setting up a hybrid communication completely in the backend with minimal user interaction~\cite{bmw}. To capture this, we introduce the hybrid node Node-15 after Node-6 and in parallel with Node-10 with attributes $L_{15} = L_{10}$, $p^0_{15} = 0.25$, and $\kappa_{15} = \kappa_{10}$. With this redesign, we obtain $L^*_{veh} = 1.7625$, which is lower than the original base loss, indicating the effectiveness of the countermeasure in improving the vehicle's security, and matching the intervention analysis discussed in Section~\ref{sec: formations}.}

\section{Conclusion and future work}
\label{sec: conlusion}

We studied the effect of network design interventions on a network of interdependent assets with varying sensitivity to investment. We first proposed an algorithm to convert large attack graphs to a reduced attack graph with fewer variables and constraints,  allowing for more computationally efficient system level loss analysis and comparison. We then considered four potential types of design interventions by a defender of this system: a series node addition, a parallel node addition, a hybrid node addition, and an additional input node. We find that added endurance (series) or informationally redundant (hybrid) components can help decrease the expected loss, while adding new features (input) or physically redundant (parallel) components in general increases expected loss, and can only be justified if there is additional stand-alone benefit to these additions. We further showcased the usability of the proposed approach by applying it on two use cases. The results recommended by this approach are close to realistic decisions taken by the manufacturers and to outcomes from studies performed by security agencies. Future work includes, \rev{generalizing the reduction algorithm to consider scenarios with insufficient budget}, and extending our proposed framework with Reinforcement Learning based techniques to study both optimal security investments and network design interventions in multi-stage attacker-defender games. 


\bibliographystyle{IEEEtran}
\bibliography{ref.bib}

\clearpage
\appendices
\label{sec: appendix}

\section{Summary of Notation}\label{app:notation}

\begin{table}[!h]
\centering
\caption{Summary of notation}
\begin{tabular}{|l|l|}
\hline
\textbf{Symbol}                    & \textbf{Description}                                                                                                                \\ \hline
$\mathcal{G}=\{\mathcal{V}, \mathcal{E}\}$             & Directed acyclic graph with nodes $\mathcal{V}$ and edges $\mathcal{E}$                                                                                                \\ \hline
$v_g$                     & \begin{tabular}[c]{@{}l@{}} Unique target node \end{tabular}                                           \\ \hline
$Post(v)$                 & \begin{tabular}[c]{@{}l@{}}Set of all nodes that can be reached from $v$\end{tabular}                                    \\ \hline
$Pre(v)$                  & \begin{tabular}[c]{@{}l@{}}Set of all the nodes from where $v$ can be reached\end{tabular}                              \\ \hline
$P_{ij}$                  & \begin{tabular}[c]{@{}l@{}}An attack path from $v_i$ to $v_j$\end{tabular}                                               \\ \hline
$\mathcal{P}_{ij}$        & \begin{tabular}[c]{@{}l@{}}Set of all paths from $v_i$ to $v_j$\end{tabular}                                             \\ \hline
$x_i$ & Security investment on node $v_i$ \\ \hline
$p_i(x_i)$                & \begin{tabular}[c]{@{}l@{}}Probability of successful attack on $v_i$ given an \\ investment of $x_i$ on it\end{tabular} \\ \hline
$\kappa_i$                & \begin{tabular}[c]{@{}l@{}}Sensitivity of node $v_i$ to investments\end{tabular}                                             \\ \hline
$p^0_i$                   & \begin{tabular}[c]{@{}l@{}}Default probability of successful attack on $v_i$\end{tabular}                                      \\ \hline
$L_i$                     & Stand-alone loss of $v_i$ if compromised                                                                                                 \\ \hline
$B$ & Security budget of the defender \\ \hline
\end{tabular}
\label{t:notation}
\end{table}

\section{Description and security attributes for the numerical experiments in Section~\ref{sec: application}}\label{app:experiment-tables}

\begin{table}[!h]
\centering
\caption{Description and security attributes for Fig. \ref{fig: scada attack graph}}
\label{table: scada description}
\begin{tabular}{|c|c|c|c|c|}
\hline
\textbf{Node} & \textbf{Description}    & $\mathbf{p^0}$ & $\mathbf{L}$ (x$10^6$)& $\mathbf{\kappa}$ \\ \hline
1    & Email host                                                                 & 0.18  & 0.01   & 1        \\ \hline
2    & Web app                                                                    & 0.18  & 0.01   & 1        \\ \hline
3    & \begin{tabular}[c]{@{}c@{}}Privilege\\ escalation\end{tabular}             & 0.09  & 0.02   & 1        \\ \hline
4    & \begin{tabular}[c]{@{}c@{}}Web app\\ host\end{tabular}                     & 0.09  & 0.02   & 1        \\ \hline
5    & \begin{tabular}[c]{@{}c@{}}Defeat DMZ\\ access control\end{tabular}        & 0.09  & 20  & 3        \\ \hline
6    & \begin{tabular}[c]{@{}c@{}}Connect to\\ firewall\end{tabular}              & 0.13  & 0.2   & 3        \\ \hline
7    & \begin{tabular}[c]{@{}c@{}}Gain access to\\ engg. workstation\end{tabular} & 0.08  & 1000  & 5        \\ \hline
8    & \begin{tabular}[c]{@{}c@{}}Gain access to\\ SCADA controls\end{tabular}    & 0.08  & 2000  & 5        \\ \hline
g    & Manipulate PLC                                                             & 1  & 10000  & -        \\ \hline
9    & Manipulate PLC                                                             & 0.07  & 50  & 5        \\ \hline
g    & Malicious actuation                                                        & -     & 100   & -        \\ \hline
\end{tabular}
\end{table}

\begin{table}[!h]
\centering
\caption{Description and security attributes for 
of the nodes and their characteristics in the attack graph illustrated in 
Fig. \ref{fig: automotive atk grf}. 
}
\label{table: autoAtkGrf description}
\begin{tabular}{lllll}
\hline
\multicolumn{1}{|l|}{\textbf{Node}} & \multicolumn{1}{l|}{\textbf{Description}}         & \multicolumn{1}{l|}{{$p^0$}} & \multicolumn{1}{l|}{$L$} & \multicolumn{1}{l|}{$\kappa$} \\ \hline
\multicolumn{1}{|l|}{1}             & \multicolumn{1}{l|}{Cellular connection}          & \multicolumn{1}{l|}{1}              & \multicolumn{1}{l|}{1}            & \multicolumn{1}{l|}{1}                 \\ \hline
\multicolumn{1}{|l|}{2}             & \multicolumn{1}{l|}{OBD-II}                       & \multicolumn{1}{l|}{0.25}              & \multicolumn{1}{l|}{1}            & \multicolumn{1}{l|}{1}                 \\ \hline
\multicolumn{1}{|l|}{3}             & \multicolumn{1}{l|}{Wi-Fi}                        & \multicolumn{1}{l|}{0.5}              & \multicolumn{1}{l|}{1}            & \multicolumn{1}{l|}{1}                 \\ \hline
\multicolumn{1}{|l|}{4}             & \multicolumn{1}{l|}{USB}                          & \multicolumn{1}{l|}{0.25}              & \multicolumn{1}{l|}{1}            & \multicolumn{1}{l|}{1}                 \\ \hline
\multicolumn{1}{|l|}{5}             & \multicolumn{1}{l|}{Connected services}           & \multicolumn{1}{l|}{0.75}           & \multicolumn{1}{l|}{5}           & \multicolumn{1}{l|}{3}                 \\ \hline
\multicolumn{1}{|l|}{6}             & \multicolumn{1}{l|}{Connection to vehicle}      & \multicolumn{1}{l|}{0.75}            & \multicolumn{1}{l|}{10}           & \multicolumn{1}{l|}{1}                 \\ \hline
\multicolumn{1}{|l|}{7}             & \multicolumn{1}{l|}{Internal network diagnostics} & \multicolumn{1}{l|}{0.75}           & \multicolumn{1}{l|}{5}           & \multicolumn{1}{l|}{3}                 \\ \hline
\multicolumn{1}{|l|}{8}             & \multicolumn{1}{l|}{Head Unit internal arch.}            & \multicolumn{1}{l|}{0.75}           & \multicolumn{1}{l|}{5}            & \multicolumn{1}{l|}{3}                 \\ \hline
\multicolumn{1}{|l|}{9}             & \multicolumn{1}{l|}{Connected services comm.}     & \multicolumn{1}{l|}{0.75}            & \multicolumn{1}{l|}{5}           & \multicolumn{1}{l|}{3}                 \\ \hline
\multicolumn{1}{|l|}{10}            & \multicolumn{1}{l|}{Remote diagnostic comm.}      & \multicolumn{1}{l|}{0.75}            & \multicolumn{1}{l|}{20}           & \multicolumn{1}{l|}{2}                 \\ \hline
\multicolumn{1}{|l|}{11}            & \multicolumn{1}{l|}{Telematics Control Unit}      & \multicolumn{1}{l|}{0.25}            & \multicolumn{1}{l|}{20}           & \multicolumn{1}{l|}{2}                 \\ \hline
\multicolumn{1}{|l|}{12}            & \multicolumn{1}{l|}{Head unit}                    & \multicolumn{1}{l|}{0.5}           & \multicolumn{1}{l|}{5}            & \multicolumn{1}{l|}{2}                 \\ \hline
\multicolumn{1}{|l|}{13}            & \multicolumn{1}{l|}{CAN tx/rx}                    & \multicolumn{1}{l|}{0.25}            & \multicolumn{1}{l|}{5}           & \multicolumn{1}{l|}{1}                 \\ \hline
\multicolumn{1}{|l|}{14}            & \multicolumn{1}{l|}{Central Gateway}              & \multicolumn{1}{l|}{1}              & \multicolumn{1}{l|}{20}           & \multicolumn{1}{l|}{1}                 \\ \hline
\multicolumn{1}{|l|}{g}             & \multicolumn{1}{l|}{Vehicle CAN}                  & \multicolumn{1}{l|}{1}              & \multicolumn{1}{l|}{50}           & \multicolumn{1}{l|}{-}                 \\ \hline
                                   &                                                   &                                     &                                   &                                       
\end{tabular}
\end{table}

\rev{\section{Obtaining optimal investments on the original graph and proof for Lemma \ref{lemma: sufficient-budget}}
\label{app: optimal_investments}

While Algorithm \ref{alg: reduction} provides the final reduced graph and the optimal expected loss. The resulting nodes may lose their physical meaning through the reduction process. It is essential for the defender to query the optimal investments on the nodes of the original attack graph. These are obtained through the reduction procedure and can be stored during each iteration of Algorithm \ref{alg: reduction}. 

From Lemma \ref{lemma:series-reduction},  given some budget $T$ invested over a series link $\{v_i \rightarrow v_{i+1} \rightarrow \ldots \rightarrow v_{i+n} \rightarrow v_{t}\}$, the optimal investment on each node is given by:
\begin{equation*}
    \begin{gathered}
        x^*_t = \frac{-1}{\kappa_t}\log\Big(\frac{L_{i+n}\kappa_{i+n}}{L_t(\kappa_t-\kappa_{i+n})}\Big), \\
        x^*_j = \frac{-1}{\kappa_j}\log\Big(\frac{L_{j-1}\kappa_{j-1}}{L_j\kappa_{j+1}}\Big(\frac{\kappa_{j+1}-\kappa_j}{\kappa_j-\kappa_{j-1}}\Big)\Big), \\
        x^*_i = T - x^*_t - \sum^{i+n}_{j=i+1} x^*_{j}, ~\forall j \in I =  \{i+1, i+2, \ldots, i+n \}.
    \end{gathered}
\end{equation*}
Similarly, from Lemma \ref{lemma: parallel}, the optimal investments on a parallel network $\{( v_i \rightarrow v_{i+1} \rightarrow v_g ), ( v_i \rightarrow v_{i+2} \rightarrow v_g ), \ldots, ( v_i \rightarrow v_{i+n} \rightarrow v_g )\}$ is given by:
\begin{equation*}
    \begin{gathered}
        x^*_j = \tfrac{-1}{\kappa_r} \log\Big(\frac{\kappa_i\kappa_{par}}{1-\kappa_i\kappa_{par}}(\frac{L_i}{L_j+L_g})\Big), \forall r = \{i+1, \ldots, i+n\} \\
        x^*_i = T - \sum^{i+n}_{r=i+1} x^*_j, ~\forall j \in I.
    \end{gathered}
\end{equation*}
And from Lemma \ref{lemma: inputs}, for a graph $\{(v_{i} \rightarrow v_{t}), (v_{i+1} \rightarrow v_{t}), \ldots, (v_{i+n} \rightarrow v_t) \}$ with $L_j = L, ~\forall j \in I$, the optimal investments on the nodes are given by:
\begin{equation*}
    \begin{gathered}
        x^*_t = \frac{-1}{\kappa_t}\log\Big(\frac{L}{L_t(\kappa_t \kappa_{par}-1)} \Big) \\
        x^*_j = \frac{1}{\kappa_j \kappa_{par}}\Big(T - x^*_t\Big), \forall j \in I.
    \end{gathered}
\end{equation*}
We see that other than the input nodes in each link, the optimal investments on the other nodes do not depend on the budget $T$.  
Hence, any increase in budget would not change the investment on these nodes. Additionally, a sufficient budget $T$ over each link would then be the minimum budget $T$ such that:
\[ T - \sum^{i+n}_{r=i+1} x^*_r \geq 0 \] 

A sequence of steps that the defender could follow to assess the utility of a design intervention would be:
\begin{enumerate}
    \item Perform a network design intervention on the original graph
    \item Run Algorithm \ref{alg: reduction} to obtain the optimal expected loss after the intervention while storing the optimal investments on the reduced nodes during each iteration.
    \item Output the optimal investments on every node.
\end{enumerate}
}

\subsection{Proof of Lemma~\ref{lemma: sufficient-budget}}
\label{app: sufficient_proof}
\rev{
\begin{proof}
The proof is by construction, and follows directly from the above arguments on finding the optimal investment profiles in each series or parallel subnetwork. 
\end{proof}
}
\rev{\section{Series reduction under insufficient budget}
\label{app:insufficient}
Under an insufficient budget, the nodes closer to the end of the link receive investment first. Investment is made sequentially, starting from the target node and moving upstream, each node receiving their optimal investments following Lemma \ref{lemma:series-reduction} until the budget runs out. This is proved in the following lemma:

\begin{lemma}
\label{lemma: nonzeroSeries}
Consider an attack graph $\mathcal{G}$ containing the series of nodes $\{v_i \rightarrow v_{i+1} \rightarrow \cdots \rightarrow v_{i+n} \rightarrow v_{i+n+1}\}$ and a budget $T$ spent over these nodes. Under a optimal investment strategy $\mathbf{x}^*$, there exits a non-zero optimal investment on $v_j, j \in \{i, i+1, i+2,\ldots i+n\}$, when:
    \begin{equation*}
    \label{Lemma: eq: budget condition}
        \sum_{v_t \in Post(v_j)} x^*_t < T 
    \end{equation*}
\end{lemma}
\begin{proof}  
Consider a series path $\mathcal{G}_t = \{v_0 \rightarrow v_1 \rightarrow v_2 \cdots \rightarrow v_t\}$ such that $\sum_{j = 1}^t x^*_j \geq T$. We prove the lemma by induction. For the base case, consider a node $v_{-1}$ added upstream to get $\mathcal{G}_1 = \{v_{-1} \rightarrow v_0 \rightarrow v_1 \rightarrow v_2 \cdots \rightarrow v_t \}$. From the discussion in Appendix \ref{app: optimal_investments} we know that the investments on the $t$ downstream nodes do not change. From Lemma \ref{lemma:series-reduction}, $x^*_{-1} = T - \sum_{j = 0}^{t}x^*_j$. Additionally, depending on the node attributes of $v_0$ and $v_{-1}$, $x^*_0 \geq 0$ which implies $\sum_{j = 0}^{t}x^*_j \geq T$ and $x^*_{-1} \leq 0$. Hence, within the feasible region the minimum occurs at $x^*_{-1} = 0$, proving the base case. Next for the inductive step, consider the path $\{v_{1-n} \rightarrow \cdots \rightarrow v_{-1} \rightarrow v_0 \rightarrow v_1 \rightarrow v_2 \cdots \rightarrow v_t\}$. We assume the statement is true for this sequence of $n+t$ nodes. With this assumption, the upstream $n$ nodes can be replaced with single node with equivalent loss $L_{eq} = \sum_{j=1-n}^0 L_j$. Finally, applying the same steps as the induction base, it is straightforward to show that the statement holds for $n+1+t$ nodes with only the $t$ nodes downstream receiving non-zero investments.
\end{proof}

\begin{conjuncture}
By viewing the entire attack graph as a series of sub-networks (series or parallel), following Lemma \ref{lemma: nonzeroSeries}, under an insufficient budget, the sub-networks closer to the target receive their respective optimal investments until the budget is depleted.    
\end{conjuncture}

\begin{conjuncture}
If the budget $T$ over a parallel link $\{( v_i \rightarrow v_{i+1} \rightarrow v_g ), ( v_i \rightarrow v_{i+2} \rightarrow v_g ), \ldots, ( v_i \rightarrow v_{i+n} \rightarrow v_g )\}$ is insufficient, the \emph{root node} $v_i$ does not receive any investment and the rest of the split across the parallel nodes $v_j$, $j \in I = \{i+1, i+2, \ldots, i+n\}$. The optimal investments on each node then being:
\begin{equation*}
    x^*_r = \frac{1}{k_rk_{par}}\Big(T + \sum_{j= I \setminus \{r\} |} \frac{1}{k_r} \log\Big(\frac{L_j + L_g}{L_r + L_g}\Big)\Big), r \in I
\end{equation*}
\end{conjuncture}

The remaining budget may or may not equalize expected losses across all paths across all parallel paths, with certain nodes receiving no investment. 
}

\clearpage
\section{Online appendix: Proofs for Section~\ref{sec: reductions}}
\label{app:reduction-proofs}

\subsection{Proof of Lemma~\ref{lemma:series_zero_investments}}
\begin{proof}
We first prove (\ref{eq: series_end_conditions}) proposed over nodes $v_{i+n-1} \rightarrow v_{i+n}$. Assume a budget of $T_n$ is spent on these nodes. Denoting the investments on these nodes by $x_{i+n-1}=T_n-x$ and $x_{i+n}=x$, the expected loss for this pair is
\[\mathbf{L}_{n} = e^{-\kappa_{i+n-1} (T_n-x)}\Big(L_{i+n-1} + L_{i+n} e^{-\kappa_{i+n}x}\Big)~.\]
The first derivative of the expected loss comes out to be
\begin{multline*}
\frac{\partial\mathbf{L}_{n}}{\partial x} = e^{-\kappa_{i+n-1} (T_n-x)}\Big(\kappa_{i+n-1}L_{i+n-1} + \\ (\kappa_{i+n-1}-\kappa_{i+n})L_{i+n} e^{-\kappa_{i+n}x}\Big)~.
\end{multline*}
The second derivative is
\begin{multline*}
    \frac{\partial^2\mathbf{L}_{n}}{\partial x^2} = e^{-\kappa_{i+n-1} (T_n-x)}\Big(\kappa_{i+n-1}^2 L_{i+n-1} + \\ (\kappa_{i+n-1} - \kappa_{i+n})^2L_j e^{-\kappa_{i+n}x}\Big)~.
\end{multline*}
Therefore in $x \in [0, T_n]$ the total loss is convex meaning the KKT conditions provide the necessary and sufficient conditions to find the unique minimizer. 

When (\ref{eq: series_end_conditions}) are satisfied, $\frac{\partial\mathbf{L}_{n}}{\partial x} \geq 0$. Hence, the minimum is obtained at the left extremum $x^* = 0$ which gives the total expected loss to be $\mathbf{L}^*_{n,1} = e^{-\kappa_{i+n-1} T}\Big(L_{i+n-1}+L_{i+n}\Big)$. This is equivalent to nodes $v_{i+n-1}$ and $v_{i+n}$ being replaced with a node $v_{eq}$ with $L_{eq} = L_{i+n-1} + L_{i+n}$ and $\kappa_{eq} = \kappa_{i+n-1}$. It now acts as the effective last node of the series connection.

When (\ref{eq: series_end_conditions}) are not satisfied i.e when $\kappa_{i+n-1} < \kappa_{i+n}$ and $L_{i+n-1} < (\frac{\kappa_{i+n}}{\kappa_{i+n-1}}-1)L_{i+n}$, $\frac{\partial\mathbf{L}}{\partial x} < 0$ if $x=0$. Hence when $x \in [0, T]$, the minimum occurs when $\frac{\partial\mathbf{L}}{\partial x} = 0$ where $x^* = x_    {i+n}^* = \frac{-1}{\kappa_{i+n}}\log\Big(\frac{\kappa_{i+n-1}}{\kappa_{i+n}-\kappa_{i+n-1}}\frac{L_{i+n-1}}{L_{i+n}}\Big)$. With $x_{i+n}^*$, the expected loss over $\{v_{i+n-2} \rightarrow v_{i+n-1} \rightarrow v_{i+n}\}$ is:
\begin{multline*}
    \mathbf{L}_{n-1} = e^{-\kappa_{i+n-2}x_{i+n-2}}\Big(L_{i+n-2} + \\ \frac{\kappa_{i+n}L_{i+n}}{\kappa_{i+n}-\kappa_{i+n-1}}e^{-\kappa_{i+n-1}x_{i+n-1}}\Big)
\end{multline*}
Similar to the above procedure, taking the first and second derivatives of $\mathbf{L}_{n-1}$, we can show that when $\kappa_{i+n-2} \geq \kappa_{i+n-1}$ or $L_{i+n-2} \geq \Big(\frac{\frac{\kappa_{i+n-1}}{\kappa_{i+n-2}} - 1}{1-\frac{\kappa_{i+n-1}}{\kappa_n}}\Big)L_{i+n-1}$, $x_{i+n-1}^* = 0$ and if this is not satisfied then $x_{i+n-1}^* \neq 0$. Similarly, using the same procedure, it can be shown that (\ref{eq: series_mid_conditions}) hold for all nodes $\{i+1, i+2, \ldots ,i+n-1\}$.
\end{proof}

\subsection{Proof of Lemma~\ref{lemma:series-reduction}}

\begin{proof}
Let $T$ denote the total budget spent at the optimal investment profile on nodes $v_j,~ \forall j\in \{i, i+1, \ldots, i+n, t\}$. {We first show that given a sequence $\{v_i\rightarrow v_j\}$, under the assumptions of the lemma, the first two nodes can be replaced by an equivalent node with $\kappa_{eq}=\kappa_i$ and $L_{eq}=\frac{\kappa_{j}}{\kappa_{j}-\kappa_i}L_i
\Big(\frac{\kappa_{j}-\kappa_i}{\kappa_{i}} \frac{L_j}{L_i}\Big)^{\frac{\kappa_i}{\kappa_j}}$. 

From the proof of Lemma \ref{lemma:series_zero_investments}, we know that the total loss is convex, and the KKT conditions are necessary and sufficient to find the (unique) minimizer. 
If $\kappa_i < \kappa_j$ and  $L_i < (\frac{\kappa_j}{\kappa_i}-1)L_j$, then the unique minimizer of the loss is $x^*=\frac{1}{\kappa_j}\log{\frac{\kappa_{j}-\kappa_i}{\kappa_{i}} \frac{L_j}{L_i}}$. The total loss under this profile is given by $\mathbf{L} = (\frac{\kappa_{j}}{\kappa_{j}-\kappa_i}L_i\Big(\frac{\kappa_{j}-\kappa_i}{\kappa_{i}} \frac{L_j}{L_i}\Big)^{\frac{\kappa_i}{\kappa_j}})e^{-\kappa_i T}$, establishing the claimed equivalence of the reduction. 

We next prove the Lemma by induction. For the base case, consider a series path $\{v_i\rightarrow v_{i+1} \rightarrow v_{i+2}\}$. We first reduce the last two nodes to an equivalent node $v_{int}$ with parameters $\kappa_{int}=\kappa_{i+1}$ and $L_{int} = \frac{\kappa_{i+2}}{\kappa_{i+2}-\kappa_{i+1}}L_{i+1}\Big(\frac{\kappa_{i+2}-\kappa_{i+1}}{\kappa_{i+1}} \frac{L_{i+2}}{L_{i+1}}\Big)^{\frac{\kappa_{i+1}}{\kappa_{i+2}}}$. We now repeat this for $\{v_i\rightarrow v_{int}\}$. This leads to an equivalent node $\kappa_{eq}=\kappa_i$ and equivalent loss
\begin{align*}
L_{eq} &= \frac{\kappa_{i+1}}{\kappa_{i+1}-\kappa_{i}}L_{i}\Big(\frac{\kappa_{i+1}-\kappa_{i}}{\kappa_{i}} \frac{L_{int}}{L_{i}}\Big)^{\frac{\kappa_{i}}{\kappa_{i+1}}}\\
& = \frac{\kappa_{i+1}L_{i}}{\kappa_{i+1}-\kappa_{i}}\Big(\frac{\kappa_{i+1}-\kappa_{i}}{\kappa_{i}}\frac{\kappa_{i+2}}{\kappa_{i+2}-\kappa_{i+1}} \frac{L_{i+1}}{L_{i}}\Big)^{\frac{\kappa_{i}}{\kappa_{i+1}}}\notag\\
& \hspace{1.8in}\Big(\frac{\kappa_{i+2}-\kappa_{i+1}}{\kappa_{i+2}}\frac{L_{i+2}}{L_{i+1}}\Big)^{\frac{\kappa_{i}}{\kappa_{i+2}}}
\end{align*}
which matches \eqref{eq:series-long}. 
Next for the inductive step, consider the path $\{v_i \rightarrow v_{i+1} \rightarrow \cdots \rightarrow v_{i+n} \rightarrow v_t\}$. We assume the statement is true for any sequence of $n$ series nodes, including $\{v_{i+1} \rightarrow v_{i+1} \rightarrow \cdots \rightarrow v_{i+n} \rightarrow v_t\}$, and replace these with a node with the equivalent loss as in \eqref{eq:series-long}. Then, applying steps similar to those in the induction base, it is straightforward to verify that \eqref{eq:series-long} holds for the $n+1$ node sequence as well.}
\end{proof}

\subsection{Proof of Lemma~\ref{lemma: non-zero parallel}}

\rev{\begin{remark}
\label{rem: optimal_strategy}
While investing on nodes in a subset of paths in $\mathcal{P}_{sg}$ would reduce their expected losses, losses along other paths may not be lower. Hence the optimal strategy would be to equalize expected losses across \emph{all} paths. If the expected losses across all paths cannot be equalized, there exist two paths $P_k$  and $P_{k'}$ with their expected losses such that $\mathbf{L}^*_k > \mathbf{L}^*_{k'}$ with multiple paths potentially having the same highest expected loss.
\end{remark}}

\begin{proof}
Assume $I = \{i+1, i+2 \ldots i+n\}$ and a budget of $T$ spent over all nodes in the parallel network. Consider the loss across some path $\{v_i \rightarrow v_j \rightarrow v_g\}$ where $j \in I$.
We begin the proof by arguing that for a set of parallel paths, the optimal investment profile \rev{looks to} equalizes the expected losses across all parallel paths. We prove this by contradiction. 

\rev{From Remark \ref{rem: optimal_strategy} we know that that an optimal strategy would be to equalize the losses across all paths.}
Assume that there exists an optimal profile $\Tilde{\mathbf{x}}^*$ that does not equalize the expected loss across all paths. 
This means there exists at least one path \rev{(potentially multiple equivalent paths)} $P_j = \{v_i \rightarrow v_{i+j} \rightarrow v_g\}$ such that $(L_i + (L_{i+j}+L_g)e^{-k_{i+j}\Tilde{x}^*_{i+j}})e^{-k_{i}\Tilde{x}^*_{i}} > (L_i + (L_{i+j}+L_g)e^{-k_{i+r}\Tilde{x}^*_{i+r}})e^{-k_{i}\Tilde{x}^*_{i}}$ \rev{for some or multiple} $~ r \in \{1,2,~...~n\}/\{j\}$. \rev{Let the expected loss under this strategy be} $\mathbf{L}_{par}(\mathbf{\Tilde{x}}^*) = ((L_{i+j}+L_g)e^{-k_{i+j}\Tilde{x}^*_{i+j}})e^{-k_{i}\Tilde{x}^*_{i}}$. Under an equalizing strategy $\mathbf{x}^*$, $\Tilde{x}^*_j < x^*_j$. This gives, $(L_i + (L_{i+j}+L_g)e^{-k_{i+j}\Tilde{x}^*_{i+j}})e^{-k_{i}\Tilde{x}^*_{i}} > (L_i + (L_{i+j}+L_g)e^{-k_{i+j}x^*_{i+j}})e^{-k_{i}x^*_{i}}$ which implies $\mathbf{L}_{par}(\Tilde{\mathbf{x}}^*) > \mathbf{L}_{par}(\mathbf{x}^*)$, contradicting the initial assumption.
If the losses across all paths are equal, we get
\[x_i = T + \sum_{k \in I \setminus \{j\}}\frac{1}{\kappa_{k}}\log\Big(\frac{L_{k}+L_g}{L_j+Lg}\Big) - \kappa_j\kappa_{par}\kappa_ix_j\]
We now optimize (\ref{eq: atkOptim}) over this path, i.e.
\[\mathbf{L}_j = \min_{x_{i}, x_j}e^{\kappa_ix_i}\Big(L_i + (L_j + L_g)e^{-\kappa_jx_j}\Big)\]
The first and second derivatives of $\mathbf{L}_j$ are
{\scriptsize{\[{\tfrac{\partial\mathbf{L}_j}{\partial x_j} = \kappa_je^{\kappa_j\kappa_{par}\kappa_ix_j}(\kappa_{par}\kappa_iL_i + (\kappa_{par}\kappa_i - 1)(L_j+L_g)e^{-\kappa_jx_j})}\]}}
{\scriptsize{\[\tfrac{\partial^2\mathbf{L}_j}{\partial x^2_j} = \kappa^2_je^{\kappa_j\kappa_{par}\kappa_ix_j}(\kappa^2_{par}\kappa^2_iL_i + (\kappa_{par}\kappa_i - 1)^2(L_j+L_g)e^{-\kappa_jx_j})\]}}
Similar to the series case, the total expected loss is convex. When $ L_i \geq \Big(\frac{1 - \kappa_i \kappa_{par}}{\kappa_i \kappa_{par}}\Big)(L_j + L_g)$, $\frac{\partial\mathbf{L}_j}{\partial x_j} \geq 0$ and the minimum occurs at $x^*_j = 0$, \rev{effectively removing it from the optimization problem}. The same reason followed for Lemma \ref{lemma:series_zero_investments} can be used to prove that $x^*_j \neq 0$ when this condition is not met.

In the general case, there may be multiple parallel nodes which satisfy these conditions. We are left to prove that it is the node with the lowest stand-alone loss ($L_{i+1}$ here) that receives no investment at the optimal investment profile. We prove this by contradiction. Assume there exists $\mathbf{\Tilde{x}}^*$ such that $x^*_j = 0$ for some $j \neq i+1$. Under such an investment strategy, the expected loss across path $P_j = (v_i \rightarrow v_j \rightarrow v_g)$ is $\Tilde{L}^*_j = \Big(L_i + L_j + L_g\Big)e^{-k_i\Tilde{x}^*_i}$. But $\Tilde{L}^*_j > \Big(L_i + (L_{i+1}+L_g)e^{-k_{i+1}\Tilde{x}^*_{i+1}}\Big)e^{-k_i\Tilde{x}^*_i}\Big)$ since $L_j > L_{i+1}$, meaning $\mathbf{\Tilde{x}}^*$ is not the optimal, which is a contradiction and the optimal strategy has $x^*_{i+1} = 0$.
\end{proof}

\subsection{Proof of Lemma~\ref{lemma: parallel}}
\begin{proof}
Assume a budget of $T$ is spent over the nodes in $\mathcal{G}$. 
Since the losses across all paths are equal, we get:
\[x_r = \Big[\frac{\kappa_n}{\kappa_r}x_n + \log\Big(\frac{L_r + L_g}{L_n + L_g}\Big)\Big] ~ \forall ~ r =\{i+1,i+2, \ldots i+n-1\}\] and 
\[x_1 = T - \Big[ 1 + \sum_{r=i+1}^{i+n-1}\frac{\kappa_n}{k_{r}}\Big]x_n -\log\Big(\frac{\prod_{r=2}^{n-1}(L_r+L_g)}{(L_n+L_g)^{n-1}}\Big) \]
The total expected loss over $\mathcal{G}$ is
\[\mathbf{L}_{par} = \min_{\mathbf{x}} \Big(L_i + (L_{i+n} + L_g)e^{-\kappa_{i+n}x_{i+n}}\Big)e^{-\kappa_ix_i}\]
It can easily be verified that $\mathbf{L}_{par}$ is convex and the KKT conditions are necessary and sufficient to find the unique minimizer. Since the conditions of Lemma \ref{lemma: non-zero parallel} do not hold for all parallel nodes, the solution to $\mathbf{L}_{par}$ is
\begin{equation*}
\label{eq: parallel investments}
    \begin{gathered}
        x^*_r = \tfrac{-1}{\kappa_r} \log\Big(\frac{\kappa_i\kappa_{par}}{1-\kappa_i\kappa_{par}}(\frac{L_i}{L_r+L_g})\Big), \forall r = \{i+1, \ldots, i+n\} \\
        x^*_i = B - \sum^{n}_{r=1} x^*_r
    \end{gathered}
\end{equation*}
and 
\begin{equation*}
\begin{split}
    \mathbf{L}^*_{par} = \frac{L_i}{1-\kappa_i \kappa_{par}}\prod_{j=i+1}^{i+n}\Big(\frac{\kappa_i\kappa_{par}}{1-\kappa_i\kappa_{par}} \Big( \frac{L_i}{L_{j}+L_g}\Big)\Big)^{\frac{-\kappa_i}{\kappa_{j}}} e^{-\kappa_iB}
\end{split}
\end{equation*}
concluding the proof.
\end{proof}

\subsection{Proof of Lemma~\ref{lemma:zero destination node}}

\begin{proof}
Assume a budget of $T$ is spent over $\mathcal{G}$. Similar to Lemma \ref{lemma: non-zero parallel}, we know that the losses across all paths are equal. Now, we get
\[x_r = \frac{\kappa_i}{\kappa_r}x_i ~\forall~ r \in \mathcal{I}\setminus \{i\}\] and
\[x_t = T - \Big(1 + \sum_{r=i+1}^{i+n}\frac{\kappa_i}{\kappa_r}\Big)x_i\]
The total expected loss over $\mathcal{G}$ is 
\[\mathbf{L}_{in} = \Big(L + L_te^{-\kappa_t x_t}\Big)e^{-\frac{(T - x_t)}{\kappa_{par}} }\]
The first and second derivative of $\mathbf{L}_{in}$ are
\[\frac{\partial\mathbf{L}_{in}}{\partial x_t} = \Big(\frac{L}{\kappa_{par}} + \Big(\frac{1}{\kappa_{par}} - \kappa_t\Big)L_te^{-\kappa_t x)_t}\Big)e^{x_t/\kappa_{par}}\]
\[\frac{\partial\mathbf{L}^2_{in}}{\partial x_t^2} = \Big(\frac{L}{\kappa_{par}^2} + \Big(\frac{1}{\kappa_{par}} - \kappa_t\Big)^2L_te^{-\kappa_t x)_t}\Big)e^{x_t/\kappa_{par}}\]
Similar to the previous cases, the loss function is convex and when $\kappa_t\kappa_{par} \leq 1$ or $L \leq (\kappa_t\kappa_{par} - 1)L_t$, the minimum is obtained at the left extremum, i.e. at $x^*_t = 0$. And similar to the previous cases, it can be proved that $x^*_t \neq 0$ if these conditions are not met.
\end{proof}

\subsection{Proof of Lemma~\ref{lemma: inputs}}
\begin{proof}
The proof follows the same steps as Lemma \ref{lemma: parallel} and is omitted here for the sake of brevity.
\end{proof}

\end{document}